\documentclass[10pt,a4]{article}

\usepackage{a4wide}
\usepackage{amsmath}
\usepackage{amsfonts}
\usepackage{amssymb}
\usepackage{amsthm}

\usepackage{graphicx}
\usepackage{graphics}

\usepackage{epsfig}

\newtheorem{theorem}{Theorem}
\newtheorem{lemma}[theorem]{Lemma}
\newtheorem{proposition}[theorem]{Proposition}
\newtheorem{observation}[theorem]{Observation}
\newtheorem{corollary}[theorem]{Corollary}
\newtheorem{definition}[theorem]{Definition}

\def\PSPACE{\mathsf{PSPACE}}

\def\NL{\mathsf{NL}}

\newcommand{\problemx}[3]{
\par\noindent\underline{\sc#1}\par\nobreak\vskip.2\baselineskip
\begingroup\clubpenalty10000\widowpenalty10000
\setbox0\hbox{\bf INPUT:\ \ }\setbox1\hbox{\bf QUESTION:\ \ }
\dimen0=\wd0\ifnum\wd1>\dimen0\dimen0=\wd1\fi
\vskip-\parskip\noindent
\hbox to\dimen0{\box0\hfil}\hangindent\dimen0\hangafter1\ignorespaces#2\par
\vskip-\parskip\noindent
\hbox to\dimen0{\box1\hfil}\hangindent\dimen0\hangafter1\ignorespaces#3\par
\endgroup}

\newcommand{\N}{\mathbb{N}}
\newcommand{\Nat}{\mathbb{N}}

\newcommand{\Z}{\mathbb{Z}}

\newcommand{\A}{\mathcal{A}}
\newcommand{\B}{\mathcal{B}}

\newcommand{\J}{\mathcal{J}}

\newcommand{\T}{\mathcal{T}}

\newcommand{\modstate}{\mathsf{Mod}}
\newcommand{\elzc}{\textsc{ZE}}

\newcommand{\trans}[1]{\stackrel{#1}{\longrightarrow}}
\newcommand{\gt}[1]{\trans{#1}}

\newcommand{\sgn}{\text{sgn}}

\newcommand{\poly}{\textsc{poly}}

\newcommand{\per}{\mathsf{per}}
\newcommand{\goto}{\mathsf{goto}}

\newcommand{\lev}{\mathsf{EqL}}
\newcommand{\eqlevel}{\mathsf{EqL}}

\newcommand{\il}{\mathsf{IL}}

\newcommand{\depictlev}[1]{\stackrel{#1}{\longleftrightarrow}}
\newcommand{\base}{\mathsf{base}}

\newcommand{\act}{\Sigma}

\newcommand{\docatrace}{\textsf{Doca-Eq}}

\newcommand{\shc}{\mathsf{SC}}
\newcommand{\stable}{\mathsf{St}}
\newcommand{\unstable}{\varepsilon}

\newcommand{\reset}{\mathsf{Res}}
\newcommand{\sink}{\mathsf{sink}}
\newcommand{\modp}{\mathsf{Mod}}
\newcommand{\trule}[1]{\stackrel{#1}{\mapsto}}

\newcommand{\fixres}{\mathsf{FixRes}}

\newcommand{\LTSext}{\T_{\mathsf{ext}}}

\newcommand{\effect}{\mathsf{effect}}

\begin{document}

\title{Equivalence of Deterministic One-Counter Automata is NL-complete}
\author{
Stanislav B\"ohm,
Stefan G\"oller,
Petr Jan\v{c}ar
}
\date{}

\maketitle
\begin{abstract}
\noindent
We prove that language equivalence of deterministic one-counter automata is 
NL-complete. This improves the superpolynomial time complexity
upper bound 
shown by Valiant and Paterson in 1975.
Our main contribution is to prove that two deterministic one-counter 
automata are inequivalent if and only if they can be distinguished by a word of length polynomial in the size of the two input automata.

\end{abstract}

\section{Introduction}\label{s:intro}

In theoretical computer science, one of the most
fundamental decision problems is the {\em equivalence problem} which asks
whether two given machines behave equivalently.
Among the various models of computation -- such as Turing machines, random access machines
and loop programs, just to mention a few of them -- the equivalence problem already becomes undecidable when one imposes
strong restrictions on their time and space consumption.

Emerging from formal language theory, a classical model of computation is that of pushdown automata.
A folklore result is that already universality (and hence equivalence) of pushdown automata is undecidable.
Concerning {\em deterministic pushdown automata (dpda)}, it is fair to say that
the computer science community knows very little about the equivalence problem and its complexity.

Oyamaguchi proved that the equivalence problem for 
real-time dpda (dpda without $\varepsilon$-transitions) is decidable \cite{Oyam87}.
It took significant further innovation to show the decidability
for general dpda, which is the
celebrated result by S{\'{e}}nizergues \cite{Seni01}, see also \cite{Seni02}. 
A couple of years later, Stirling showed that dpda equivalence is in fact primitive
recursive \cite{Stir02}, and his bound is still the best known upper bound for this problem.
Probably due to its intricacy, this fundamental problem has not attracted too
much research in the past ten years;
only recently a simplified proof has been announced \cite{Jan12}, with 
no substantial improvement of the complexity bound.

It is burdensome to realize that
for equivalence of dpda there is still an enormous 
complexity gap, where the mentioned 
upper bound is far from 
the best known lower bound, i.e. from
$\mathsf{P}$-hardness (which straightforwardly follows from
$\mathsf{P}$-hardness of the emptiness problem).

The same complexity gap persists even for real-time dpda.
Thus, further subclasses of dpda have been studied.
A $\mathsf{coNP}$ upper bound is known \cite{Seni03} for {\em finite-turn dpda} which
are dpda where the number of switches between pushing and popping phases is bounded.
For {\em simple dpda} (which are single state and real-time dpda),
equivalence is even decidable in polynomial time \cite{HiJeMo96}
(see~\cite{ConLasota10} for the currently best known upper bound).

{\em Deterministic one-counter automata (doca)} are one of
the simplest infinite-state computational models, extending deterministic finite
automata just with one nonnegative integer counter;
doca are thus dpda over a singleton stack
alphabet plus a bottom stack symbol.
Doca were first studied by Valiant and
Paterson in 1975 \cite{VP75}; they showed
that equivalence
is decidable in time $2^{O(\sqrt{n\log n})}$, and
a simple analysis of their proof reveals that the equivalence problem is in $\PSPACE$.
The problem is easily shown to be  $\mathsf{NL}$-hard, there is
however
an exponential gap 
between $\mathsf{NL}$ and $\mathsf{PSPACE}$.
There were attempts to settle the complexity 
of the doca equivalence problem (later we mention some) but 
the problem proved intricate; only recently 
$\mathsf{NL}$-completeness was established 
for real-time one-counter automata~\cite{BG11} but it was far from clear
if and how the proof can be extended to the general case.

Let us mention that a convenient and equi-succinct way to present a doca 
is to partition the control states (and thus the configurations)
into \emph{stable states}, in which the automaton waits for a letter to be read,
and into \emph{reset states}, in which the counter is reset
to zero and the residue class of the current counter value modulo some specified number
 determines the successor (stable) state.
Technically speaking, the difference between deterministic one-counter automata and 
their real-time variant is the lack of reset states in the real-time case.
The presence of reset states substantially increases the difficulty of the
equivalence problem.

One reason seems to be that a doca can exhibit a behaviour with
exponential periodicity, demonstrated by the following example 
(which slightly adapts the version from~\cite{VP75}).
We take a family
$(\A_n)_{n\geq 1}$ where $\A_n$
is a doca 
accepting the regular language 
$L_n=\{a^mb_i\mid 1\leq i\leq n, m\equiv 0\ (\text{mod }p_i)\}$, where $p_i$ 
denotes the $i^{\text{th}}$ prime number.
The index of the Myhill-Nerode congruence of $L_n$ is obviously
$2^{\Omega(n)}$ but we can easily construct $\A_n$
with $O(n^2\log n)$ states.
The example also demonstrates that doca are 
exponentially more succint than
their real-time variant, since one can prove that
real-time deterministic one-counter automata accepting $L_n$ 
have $2^{\Omega(n)}$ states. 
It is also easy to show that 
doca are strictly more expressive than their
real-time variant.
Analogous expressiveness and succinctness results hold for dpda and real-time
dpda, respectively.

As mentioned above, this increase in difficulty in the presence of $\varepsilon$-transitions
is confirmed by the fact that it took more than a decade to lift the decidability of real-time dpda \cite{Oyam87} to the general case \cite{Seni01,Seni02}.

{\bf Our contribution and overview.} 
The main result of this paper is that equivalence
of doca is $\mathsf{NL}$-complete,
thus closing the exponential complexity gap that has been existing for 
over thirty-five years ever since doca were introduced.

The above-mentioned exponential behavior of doca is 
reflected in our central notion of 
\emph{extended deterministic transition system} $\LTSext(\A)$ that is attached to each doca $\A$.
This system
includes a special finite deterministic transition system
which might be exponentially large in the size of
$\A$
and which corresponds to the \emph{special-mode} variant of 
stable configurations. Roughly speaking, in the special mode
we do not count with reaching the zero value in the counter 
unless a reset state is visited, and each reset-state visit 
finishes the special mode. 
Hence the special mode assumes that the counter is positive
and it only requires to remember 
finite information which is sufficient to perform the resets
correctly; in more detail, only the current control state
and the current residue classes of the counter value
w.r.t. the numbers associated with reset states are needed.

For understanding the shortest words distinguishing 
two stable inequivalent configurations of $\A$, it turns out useful 
to include also the special-mode variants of the configurations
in the study.
This allows us to show that \emph{shortest distinguishing
words} for two zero configurations have \emph{polynomial length}.

In Section~\ref{s:definitions}  we introduce basic definitions and state our main
result that equivalence of doca is $\mathsf{NL}$-complete.
A proof of the central claim on polynomial length
is given 
in Section 3 which is in turn divided into the following parts.
We give a brief overview of shortest positive paths in the transition system of
a doca in Section \ref{sub:pospaths}; this is the only part which is
derived directly from~\cite{VP75}.
In Section \ref{sub:extdoca} 
we introduce the above mentioned central notion
$\LTSext(\A)$,
and we make a straightforward analysis of some useful related notions
in Sections~\ref{sub:elzc}--\ref{sub:eq-lines}.
In particular, in Section \ref{sub:il} we study 
the \emph{independence level} of a configuration,
as the length of a shortest distinguishing word for the configuration and
its special-mode variant.
This allows us to make various useful observations, 
e.g. about \emph{linear relations} between counter values 
of configurations with the same independence level 
in Section~\ref{sub:eq-lines}.

Sections \ref{sub:line-climbing} and \ref{subsec:gapssmall} contain
the main argument. Sections \ref{sub:line-climbing} shows that when 
following a shortest distinguishing word for two zero configurations,
we cannot get a long \emph{line-climbing} segment
in which the counter values grow at both
sides, keeping a linear relation entailed by keeping the same
independence levels. 
Section~\ref{subsec:gapssmall} then shows that 
a shortest distinguishing word
for two zero configurations
cannot be long without having a long line-climbing segment. 

In Section~\ref{S App} we add a remark on the \emph{regularity
problem}. In Appendix we sketch the standard ideas of showing that 
the deterministic one-counter automata as introduced in \cite{VP75}
and the above-mentioned reset model that we work
with are equi-succinct. We also make clear that 
our simple form of
language equivalence, called \emph{trace equivalence}, 
does not bring any loss of
generality.

{\bf Related work.}
As mentioned above, doca were introduced by Valiant and
Paterson in \cite{VP75}, where
the above-mentioned 
$2^{O(\sqrt{n\log n})}$ time upper bound
for language equivalence was proven.
Polynomial time algorithms for language equivalence and
inclusion for strict subclasses of
doca were provided in \cite{HWT95,HWT98}.
In \cite{BeRo87,FRo95} polynomial time learning algorithms were presented for doca.
Simulation and bisimulation problems on one-counter automata
were studied in \cite{BGJ10,JKM00,JMS99,Mayr03}.
In recent years one-counter automata have attracted a lot of attention 
in the context of formal verification
\cite{HKOW-09concur,GoLo10,GHOW10,GoMaTo09}.

\smallskip
\noindent
\emph{Remark}:
In \cite{BeRo87,Roos88} it is stated that equivalence of doca can 
be decided in polynomial time. 
Unfortunately, 
the proofs provided in \cite{BeRo87,Roos88} 
were not exact enough to be verified,
and they raise several questions which are unanswered to date.

\section{Definitions and results}\label{s:definitions}

By $\N$ we denote the set  $\{0,1,2,\ldots\}$
of non-negative integers, and by $\Z$
the set of all integers.
For a finite set $X$, by $|X|$ we denote its cardinality.

By $\Sigma^*$  we denote the set of finite sequences of
elements of $\Sigma$, i.e. of \emph{words} over $\Sigma$.
For $w\in\Sigma^*$, $|w|$ denotes the length of $w$.
By $\varepsilon$ we denote
the empty word; hence $|\varepsilon|=0$.
If $w=uv$ then $u$ is a \emph{prefix} of $w$ and $v$ is a
\emph{suffix} of $w$.

By $\div$ we denote integer division; for $m,n\in\N$
where $n>0$ we have $m=(m\div n)\cdot n + (m\bmod n)$. 
We use \,``$\bmod$'' in two ways, clarified by the following example: 
$3 = 18 \bmod 5$, $8\neq 18\bmod 5$, $3\equiv 18\, (\bmod\, 5)$, 
$8\equiv 18\, (\bmod\, 5)$.
For $m\in\Z$, $|m|$ denotes the absolute value of $m$.

We use $\omega$ to stand for infinity;
we stipulate
$z<\omega$ and $\omega+z=z+\omega=\omega$ for all $z\in\Z$.

\medskip

A {\em deterministic labelled transition system}, a \emph{det-LTS} 
for short,
is a tuple
\begin{center}
$\T=(S_\stable, S_\unstable,\Sigma,
(\trule{a})_{a\in\Sigma},\trule{\varepsilon})$
\end{center}
where
$S_\stable$ and  $S_\unstable$ are 
(maybe infinite) disjoint sets of {\em stable states} and 
{\em unstable states}, respectively,
$\Sigma$ is
a finite \emph{alphabet},
$\trule{a}\subseteq S_\stable\times (S_\stable\cup S_\unstable)$,
for $a\in\Sigma$, and 
$\trule{\varepsilon}\subseteq 
S_\unstable\times S_\stable$ are sets of 
\emph{labelled transitions}; 
for each $s\in S_\unstable$ there is precisely one 
$t\in S_\stable$ such that $s\trule{\varepsilon}t$, whereas
for any $s\in S_\stable$ and $a\in\Sigma$ there is at most 
one $t\in S_\stable\cup S_\unstable$
such that $s\trule{a}t$.
 For all $w\in\Sigma^*$, we define relations 
$\trans{w}\subseteq S\times S$,
where $S = S_\stable\cup S_\unstable$,
inductively:
$s\trans{\varepsilon}s$ for each $s\in S$; 
if $s\trule{\varepsilon}t$ then $s\gt{\varepsilon}t$;
if $s\trule{a}t$ ($a\in\Sigma$) then 
$s\gt{a}t$;
if $s\gt{u}s'$ and $s'\gt{v}t$ 
($u,v\in\Sigma^*$) then $s\gt{uv}t$.

By $s\trans{w}$ we denote that $w$ is \emph{enabled
in} $s$, 
i.e. $s\trans{w}t$ for some $t$.

\medskip

Given  $\T=(S_\stable, S_\unstable,\Sigma,
(\trule{a})_{a\in\Sigma},\trule{\varepsilon})$,
\emph{trace equivalence} $\sim$ on $S=S_\stable\cup S_\unstable$ 
is defined as follows:
\begin{center}
$s\sim t\quad$ if
$\quad\forall w\in\act^*: s\gt{w} \,\Leftrightarrow\, t\gt{w}$.
\end{center}
Hence two states are equivalent iff they enable the same set of words
(also called traces).
A word $w\in\Sigma^*$ is a \emph{non-equivalence witness}  \emph{for} $(s,t)$,
a \emph{witness for $(s,t)$} for short,  
if $w$ is enabled in precisely one of $s,t$. 

\medskip

\emph{Remark.}
By the above definitions,
$s\trule{\varepsilon}t$ implies $s\sim t$.
This could suggest merging the states $s$ and $t$
but
we keep them separate since this is convenient
in the definitions of det-LTSs generated by
deterministic one-counter automata, as given below.

\medskip

We put 
$\act^{\leq i}=\{w\in\act^*;\, |w|\leq i\}$,
and 
we note that $\sim\,=\bigcap\,\{\sim_i\mid i\in\Nat\}$ where 
the equivalences
$\sim_0\,\supseteq\, \sim_1\,\supseteq\,\sim_2\,\supseteq\dots$ are
defined as follows:
\begin{center}
$s\sim_i t\quad$ if 
$\quad\forall w\in\act^{\leq i}: s\gt{w} \,\Leftrightarrow\, t\gt{w}$.
\end{center}
Each pair of states $(s,t)$ has the
\emph{equivalence level}, 
the \emph{eqlevel} for short,
$\eqlevel(s,t)\in\N\cup\{\omega\}$:

\begin{center}
$\eqlevel(s,t)=\begin{cases}
\omega & \text{ if $s\sim t$},\\
\max\{j\in\N\mid s\sim_j t\} & \text{ otherwise.}
\end{cases}
$
\end{center}
We also write  $s\depictlev{e}t$ instead of
$\eqlevel(s,t)=e$ (where $e\in\N\cup\{\omega\}$).
We note that the length of any \emph{shortest} witness for $(s,t)$,
where $s\not\sim t$, is  $\lev(s,t)+1$.
We also highlight the next simple fact (valid since 
our LTSs are \emph{deterministic}).

\begin{observation}\label{obs:eqleveldrop}
Suppose $s\gt{w}s'$ and $t\gt{w}t'$ in a given det-LTS. Then we have:
\begin{enumerate}
\item
$\lev(s',t')\geq \lev(s,t)-|w|$. (Hence $s'\sim t'$ if $s\sim t$.)
\item
If $w$ is a (proper) prefix of a witness for $(s,t)$ then 
$\lev(s',t')= \lev(s,t)-|w|$.
\end{enumerate}
\end{observation}

\medskip

 A {\em deterministic one-counter automaton},
 a \emph{doca} for short,
 is a tuple
\begin{center} 
$\A=(Q_\stable,Q_\reset,\Sigma,\delta,(\per_s)_{s\in Q_\reset},
(\goto_s)_{s\in Q_\reset})$
\end{center}
where 
$Q_{\stable}$
and $Q_{\reset}$
are disjoint finite sets of 
 \emph{stable control states} and
\emph{reset control states}, respectively,
$\Sigma$ is a finite \emph{alphabet}, 
$\delta\subseteq Q_\stable\times \Sigma
\times\{0,1\}\times (Q_\stable\cup Q_\reset)\times \{-1,0,1\}$ is
a set of \emph{(transition) rules},
 $\per_s\in\N$ are \emph{periods} satisfying
 $1\leq\per_s\leq |Q_\stable|$,
and $\goto_s:\{0,1,2,\dots,\per_s{-}1\}\rightarrow Q_{\stable}$ are
\emph{reset mappings}.   
For each 
$p\in Q_\stable$, $a\in\Sigma$, $c\in\{0,1\}$ there is at most one
pair $(q,j)$
(where $q\in Q_\stable\cup Q_\reset$, $j\in\{{-}1,0,1\}$) such
that $(p,a,c,q,j)\in\delta$; moreover, if $c=0$ then $j\neq {-}1$.
The tuples $(p,a,0,q,j)\in\delta$ are called the \emph{zero rules},
the tuples  $(p,a,1,q,j)\in\delta$ are the \emph{positive rules}.

\medskip

A doca $\A=(Q_\stable,Q_\reset,\Sigma,\delta,(\per_s)_{s\in Q_\reset},
(\goto_s)_{s\in Q_\reset})$
defines the det-LTS 
\begin{equation}\label{eq:TA}
\T(\A)=(Q_\stable\times\N,Q_\reset\times\N,\Sigma,
(\trule{a})_{a\in\Sigma},\trule{\varepsilon})
\end{equation}
where  $\trule{a}$ and $\trule{\varepsilon}$ are 
defined by the following (deduction) rules.
\begin{enumerate}
\item
If $(p,a,1,q,j)\in\delta$ and $n>0$ then 
$(p,n)\trule{a}(q,n{+}j)$.
\item
If $(p,a,0,q,j)\in\delta$ then 
$(p,0)\trule{a}(q,j)$.
(Recall that $j\in\{0,1\}$ in this case.)
\item
If $s\in Q_\reset$ and $n\geq 0$
then $(s,n)\trule{\varepsilon}(q,0)$ where 
$q=\goto_s(n\bmod \per_s)$. 
\end{enumerate}
An example of a doca with the respective det-LTS 
is sketched in Fig.~\ref{fig:doca-example}.

By a \emph{configuration} $C$ of the  doca $\A$ we mean $(p,m)$, 
usually written as $p(m)$, where $p$ is its
control state and $m\in\N$ 
is its \emph{counter value}.
If $C=p(0)$ then it is a \emph{zero configuration}. 
If $p\in Q_\stable$ then $C=p(m)$ is a
\emph{stable configuration}; if $p\in Q_\reset$ then $p(m)$ is 
a \emph{reset configuration}.

\begin{figure}[ht]
\centering
\epsfig{scale=0.25,file=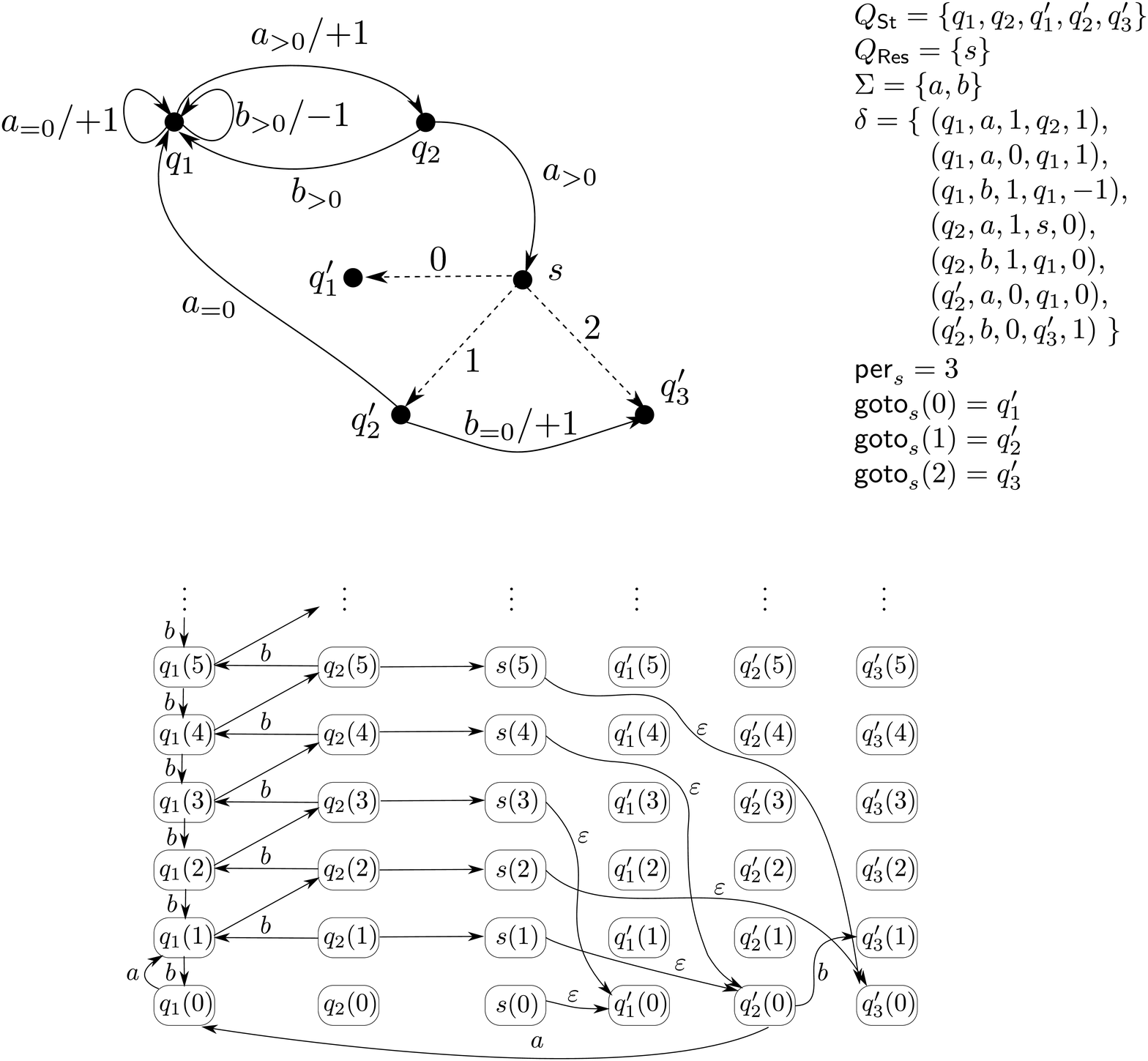}
\caption{
A doca $\A$, presented by a graph, and a fragment of $\T(\A)$}\label{fig:doca-example}
\end{figure}

The definition of (general) det-LTSs induces
the relations $\gt{w}$ ($w\in\Sigma^*$) on 
$Q\times\N$ where $Q=Q_\stable\cup Q_\reset$.
We are interested in the  
\emph{doca equivalence problem},
denoted

\medskip

$\docatrace$:
\begin{quote}
\emph{Instance}: A doca $\A$
and two stable zero
configurations 
$p(0), q(0)$.

\emph{Question}: Is $p(0)\sim q(0)$ in $\T(\A)$~?
\end{quote}

Our main aim is to show the following theorem.

\begin{theorem}\label{th:polywitness}
There is a polynomial $\poly:\N\rightarrow\N$ 
such that for any 
$\docatrace$ instance $\A, p(0),q(0)$ where
$\A$ has $k$ control states
we have that
$p(0)\not\sim q(0)$ implies
$\lev(p(0),q(0))\leq \poly(k)$. 
\end{theorem}

Using Theorem~\ref{th:polywitness}, we easily get the next theorem.
\begin{theorem}\label{th:nlcomplete}
$\docatrace$ 
is $\NL$-complete.
\end{theorem}

\begin{proof}
The lower bound follows easily from
$\NL$-hardness of digraph reachability.

On the other hand,
given a $\docatrace$ instance $\A, p(0),q(0)$, a nondeterministic
algorithm can perform the phases $j=0,1,2,\dots$ described as follows.
In phase $j$, there is a pair
$(p_j(m_j),q_j(n_j))$
in memory, the counter values $m_j, n_j$ written in binary;
for $j=0$ we have $(p_j(m_j),q_j(n_j))=(p(0),q(0))$.
If $\lev(p_j(m_j),q_j(n_j))>0$ then
a letter $a$ is nondeterministically chosen, and 
$(p_j(m_j),q_j(n_j))$ is replaced with 
$(p_{j+1}(m_{j+1}),q_{j+1}(n_{j+1}))$ where 
$p_j(m_j)\gt{a}p_{j+1}(m_{j+1})$
and $q_j(m_j)\gt{a}q_{j+1}(m_{j+1})$.

If $p(0)\not\sim q(0)$ then
Theorem~\ref{th:polywitness} guarantees that a pair 
$(p_j(m_j),q_j(n_j))$ with $\lev(p_j(m_j),q_j(n_j))=0$ can 
be thus reached 
by using only logarithmic space.

Hence  $\docatrace$ is in co-$\NL$. 
Since $\NL$$=$co-$\NL$, we are done.
\end{proof}

\section{Proof of Theorem 1}\label{s:theorem1}

\emph{Convention.}
When considering a doca $\A$, we will always tacitly assume 
the notation
\begin{equation}\label{eq:doca}
\A=(Q_\stable,Q_\reset,\Sigma,\delta,(\per_s)_{s\in Q_\reset},
(\goto_s)_{s\in Q_\reset})
\end{equation}
if not said otherwise. 
We also reserve $k$ for denoting the number of control states, i.e.
\begin{center}
$k=|Q_\stable|+|Q_\reset|$.
\end{center}
To be more concise
in the later reasoning concerning a given doca $\A$,
we use the words ``few'', ``small'', or ``short''
when we mean that the relevant quantity is bounded 
by a polynomial in $k$; 
the polynomial is always
independent of $\A$.
By a small rational number we mean $\rho=\frac{a}{b}$
or $\rho=-\frac{a}{b}$ where 
$a,b\in\N$ are small.
We also say that
\begin{center}
\emph{a set is small if its cardinality is a small number}.
\end{center}
We note that if all elements of a set 
$X$ of (integer or rational) numbers
are small
then $X$ is a small set; the opposite is not true in general.
 We often tacitly use the fact that 
\begin{center}
\emph{a quantity arising as the sum or the product 
of two small
quantities is also small}.
\end{center}
Though these expressions might look informal, they can be always
easily 
replaced by the formal statements which they abridge.
By this convention, 
Theorem~\ref{th:polywitness} says that the eqlevel of
any pair of zero configurations is small when finite.

\emph{Remark.}
It will be always obvious that we could calculate a concrete
respective polynomial whenever we use 
``few'', ``small'',
``short'' in our claims.
But such calculations would add tedious technicalities, 
and they would be not particularly rewarding w.r.t. the degree of
the polynomials. We thus prefer a transparent concise proof which
avoids technicalities whenever possible.

\subsection{Shortest positive paths in {\large {\bf$\T(\A)$}}}\label{sub:pospaths}

We first define the notion of paths in general det-LTSs, 
and then we look at special paths in $\T(\A)$, for a doca $\A$.

\begin{definition}\label{def:paths}
Given a det-LTS $\T=(S_\stable, S_\unstable,\Sigma,
(\trule{a})_{a\in\Sigma},\trule{\varepsilon})$,
a \emph{path} in $\T$ is a sequence 
\begin{center}
$s_0\gt{a_1}s_1\gt{a_2}\dots\gt{a_z}s_z$ ($z\in\N$)
\end{center}
 where $s_i\in S_\stable$ and 
$a_i\in\Sigma$
(for all $i, 0\leq i\leq z$);
it is a path \emph{from its start}
$s_0$ \emph{to its end} $s_z$.
For any $i_1,i_2$, where $0\leq i_1\leq i_2\leq z$,
the sequence $s_{i_1}\gt{a_{i_1+1}}s_{i_1+1}\gt{a_{i_1+2}}\cdots
\gt{a_{i_2}}s_{i_2}$ 
is a \emph{subpath} of the above path.
Slightly abusing notation, we will also use $s\gt{w}$ and
$s\gt{w}t$ ($s,t\in S_\stable$) to denote paths.

We also refer to $s\gt{a}t$ where $s,t\in S_\stable$ and $a\in \Sigma$
as to a \emph{step}. 
If $s\trule{a}t$ then it is  a \emph{simple step};
if $s\trule{a}s'\trule{\varepsilon}t$ then it is 
 a  \emph{combined step}. 
 The \emph{length of a path} $s_0\gt{a_1}s_1\gt{a_2}\dots\gt{a_z}s_z$
 is $z$, i.e.   the number of its steps.
\end{definition}

\medskip

When discussing the det-LTS $\T(\A)$ for a doca $\A$,
we use the term \emph{reset steps} instead of
combined steps. We now concentrate on positive paths in $\T(\A)$,
defined as follows.

\begin{definition}
Given a doca $\A$ (in notation~(\ref{eq:doca})),
a \emph{path} 
\begin{equation}\label{e:path}
p_0(m_0)\gt{a_1} p_1(m_1)\gt{a_2}\cdots \gt{a_z}p_z(m_z)
\end{equation}
in $\T(\A)$ is \emph{positive} if each step
$p_i(m_i)\gt{a_{i+1}}p_{i+1}(m_{i+1})$ ($0\leq i<z$)
is simple and is induced by a positive rule
$(p_i,a_{i+1},1,p_{i+1},j)\in \delta$ (where $j=m_{i+1}{-}m_i$).

The \emph{effect} (or the \emph{counter change}) 
\emph{of the path} (\ref{e:path}) is
$m_z{-}m_0$; if the path is positive, its effect is
an integer in the interval $[{-}z,z]$.
The path~(\ref{e:path}) is a \emph{control state cycle} if 
it is positive and we have $z>0$ and $p_z=p_0$.
\end{definition}
We note that if (\ref{e:path}) is positive then there is 
no reset step in the path and
$m_i>0$ for all $i, 0\leq i<z$; 
but we can have $m_z=0$.

The next lemma can be easily derived from Lemma $2$ in~\cite{VP75};
we thus only sketch the idea.
The claim of the lemma is illustrated
in Fig.~\ref{fig:valiant}.
\begin{figure}[ht]
\centering
\epsfig{scale=0.2,file=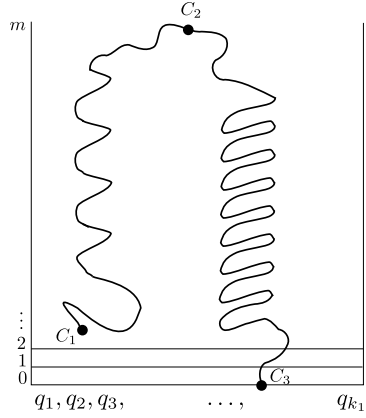}
\caption{Shortest positive paths in $\T(\A)$,
one from a configuration $C_1$ to $C_2$
and one from $C_2$ to a zero configuration $C_3$.
(Only the stable control states $q_1,q_2,\dots,q_{k_1}$ are depicted.)}\label{fig:valiant}
\end{figure}

\begin{lemma}\label{lem:valiantpositpath}
If there is a positive path 
from $p(m)$ to $q(n)$
in $\T(\A)$ then some of the shortest positive paths 
from $p(m)$ to $q(n)$
is of the form  
\begin{center}
$p(m)\gt{u_1}p'(m')\gt{v^i}p'(m'{+}id)\gt{u_2}q(n)$
\end{center}
where  $u_1$ is a short word, called the \emph{pre-phase}, 
$p'(m')\gt{v}p'(m'{+}d)$ is a short control state cycle with 
the
effect $d\in\Z$,
and $u_2$ is a short word, called the \emph{post-phase}. 
(The cycle $v$ is repeated $i$ times, where $i\geq 0$.)
\end{lemma}

\begin{proof} (Sketch.)
Lemma $2$ in~\cite{VP75} considers the case when $m\geq n+k^2$.
The cycle $v$ shown by that lemma has the length in $\{1,2,\dots,k\}$
and the effect in 
$\{-1,-2,\dots,-k\}$. The length of the pre-phase plus the post-phase
is bounded by $k^2$.
The idea is to use a most effective control state cycle for
repeating (with the largest ratio $\frac{|\text{effect}|}{\text{length}}$),
and to add the ``cost'' of reaching that cycle from $p(m)$
and of reaching $q(n)$ from the end of
the repeated cycle. The technical details can be found in~\cite{VP75}.

The situation with $n\geq m+k^2$ is handled symmetrically.
Having solved the case $|n-m|\geq k^2$, the case  
$|n-m|< k^2$ is obvious, as can be seen in Fig.~\ref{fig:valiant}:
if a long path is going up via a short cycle with a positive effect
$d_1$
and then down via another short cycle with a negative effect $-d_2$,
then it can be shortened by removing $d_2$ copies of the first cycle
and $d_1$ copies of the second cycle.
Hence $|n-m|< k^2$ implies that there is a short positive path 
$p(m)\gt{u_1}q(n)$ (with $v=u_2=\varepsilon$).
\end{proof}

It is useful to highlight the following corollary of the previous lemma.

\begin{corollary}\label{cor:valiantpositpath}
If $|m-n|$ is small and there is 
a positive path from $p(m)$ to $q(n)$ then there is a short 
positive path from $p(m)$ to $q(n)$.
\end{corollary}

\subsection{The extended det-LTS {\large $\LTSext(\A)$}}\label{sub:extdoca}

We now introduce a central notion, 
the det-LTS $\LTSext(\A)$,
which extends the det-LTS $\T(\A)$ defined in~(\ref{eq:TA}),
for a given doca
 $\A=(Q_\stable,Q_\reset,\Sigma,\delta,(\per_s)_{s\in Q_\reset},
(\goto_s)_{s\in Q_\reset})$.

Before giving a formal definition, 
we give an intuitive explanation.
Let us (temporarily) imagine that $\A$ has also a \emph{special mode of behaviour},
besides
the \emph{normal mode} defined previously; 
let any configuration $p(m)$ have
its special-mode analogue $\overline{p}(m)$.
For any \emph{positive} counter value $m>0$,
each transition $p(m)\trule{a}q(m{+}j)$ ($a\in\Sigma$)
induces the transition $\overline{p}(m)\trule{a}\overline{q}(m{+}j)$.
Further,
any transition 
$s(m)\trule{\varepsilon}q(0)$ ($s\in Q_\reset$, $m\geq 0$) induces  
$\overline{s}(m)\trule{\varepsilon}q(0)$; hence 
the special mode is finished by
any reset step, after which 
the normal mode applies.
\emph{A crucial property of the special mode}
is that whenever a configuration
$\overline{p}(0)$, where $p\in Q_\stable$,
is entered 
(by a non-reset step), a multiple 
(the
least common multiple, say) $\Delta\in\Nat$ 
of all periods $\per_s$, $s\in Q_\reset$,
is silently added to the counter
(we put $\Delta=1$ when $Q_\reset=\emptyset$). 
Hence the zero rules are never used
in the special mode 
since the counter is always positive (until a possible reset step is
performed). If we added the special-mode configurations and the
respective transitions to $\T(\A)$, we would easily observe
that
\begin{itemize}
\item
$p(m)\sim_m \overline{p}(m)$ (thus $\lev(p(m),\overline{p}(m))\geq
m$);
\item
$p(m)\not\sim \overline{p}(m)$ iff there is a positive path 
$p(m)\gt{u}q(0)$ (and thus
$\overline{p}(m)\gt{u}\overline{q}(0)=\overline{q}(\Delta)$)
\\
for some $q\in Q_\stable$
such that $q(0)\not\sim \overline{q}(0)$ 
(i.e. $q(0)\not\sim \overline{q}(\Delta)$);
\item
if $m\equiv m'\,(\bmod\, \per_s)$ for all $s\in Q_\reset$
then 
$\overline{p}(m)\sim\overline{p}(m')$;
\item
if $s\in Q_\reset$ and 
$m\equiv m'\,(\bmod\, \per_s)$ then
$\overline{s}(m)\sim\overline{s}(m')$.
\end{itemize}
In the special mode of $\A$,
the concrete value $m$ of the counter is not important once we know 
the tuple 
$(c_s)_{s\in Q_\reset}$ where $c_s=m\bmod \per_s$;
in a reset configuration $\overline{s}(m)$, 
knowing just $c=m\bmod \per_s$ is sufficient.

We do not formalize the above notions and claims, since they only 
serve us for a better understanding of the definition of  $\LTSext(\A)$
given below. The det-LTS $\LTSext(\A)$ arises from $\T(\A)$ 
by adding a finite set $Q_\modp$ of stable states
and a finite set $Q_\fixres$ of unstable states and the transitions
defined below.
The transitions from $Q_\modp$ will only lead to 
$Q_\modp\cup Q_\fixres$, whereas the $\varepsilon$-transitions from 
$Q_\fixres$ lead to zero configurations in $\T(\A)$. 
There are no 
transitions leading 
from the configurations in $\T(\A)$ to $Q_\modp\cup Q_\fixres$,
and the subgraph of $\LTSext(\A)$ arising by the restriction 
to the configurations of
$\T(\A)$ is $\T(\A)$ itself.
We thus also safely use the same symbols $\trule{a}$,
$\trule{\varepsilon}$ in both $\T(\A)$ and $\LTSext(\A)$. 
An example is sketched in Fig.~\ref{fig:texta}.

\begin{figure}[ht]
\centering
\epsfig{scale=0.4,file=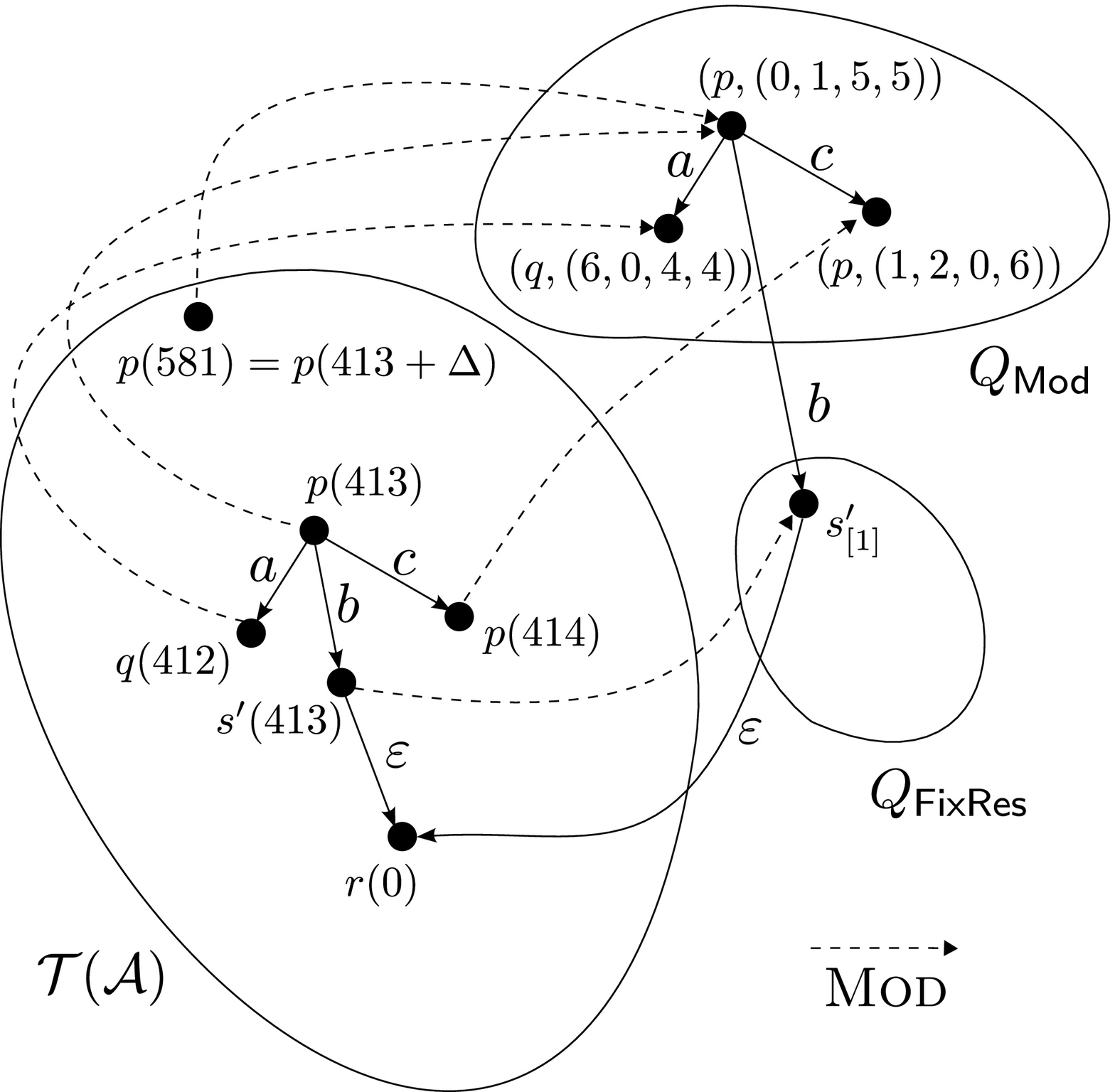}
\caption{
A fragment of $\LTSext(\A)$ where: $\{p, q, r\} \subseteq Q_\stable$, 
$Q_\reset=\{s, s', s'', s'''\}$,
$\{a, b, c\} \subseteq \Sigma$,
$\{ (p, a, 1, q, -1), (p, b, 1, s', 0), (p, c, 1, p, 1) \} \subseteq \delta$,
$(\per_s, \per_{s'}, \per_{s''}, \per_{s'''}) = (7,4,6,8)$, $\goto_{s'}(1) = r$, $\Delta = lcm\{7,4,6,8\} = 168$. 
}\label{fig:texta}
\end{figure}

\begin{definition}\label{def:extTA}
Given a doca $\A=(Q_\stable,Q_\reset,\Sigma,\delta,(\per_s)_{s\in Q_\reset},
(\goto_s)_{s\in Q_\reset})$, with the associated det-LTS 
\ $\T(\A)=(Q_\stable\times\N,Q_\reset\times\N,\Sigma,
(\trule{a})_{a\in\Sigma},\trule{\varepsilon})$, we define the
det-LTS
\begin{center}
$\LTSext(\A)=
((Q_\stable\times\N)\cup Q_\modp,
(Q_\reset\times\N)\cup Q_\fixres,\Sigma,
(\trule{a})_{a\in\Sigma},\trule{\varepsilon})$ 
\end{center}
as the extension of $\T(\A)$ where 
\begin{itemize}
\item
$Q_\modp=\{(p,(c_s)_{s\in Q_\reset})\mid p\in Q_\stable, 0\leq
c_s\leq\per_s-1\}$,
\item
$Q_\fixres=\{ s_{[c]}\mid s\in Q_\reset, 0\leq c\leq \per_s-1\}$,
and
\item
the  additional transitions
are defined by the following (deduction) rules:
\begin{enumerate}
\item
If $(p,a,1,q,j)\in\delta$ and $q\in Q_\stable$
then for each 
$(p,(c_s)_{s\in Q_\reset})\in Q_\modp$
we have 
\begin{center}
$(p,(c_s)_{s\in Q_\reset})\trule{a}(q,(c'_s)_{s\in Q_\reset})$
\end{center}
where 
$c'_s=(c_s{+}j)\bmod\per_s$ for each $s\in Q_\reset$.
\item
If $(p,a,1,s',j)\in\delta$ and  $s'\in Q_\reset$ then 
for each
$(p,(c_s)_{s\in Q_\reset})\in Q_\modp$
we have 
\begin{center}
$(p,(c_s)_{s\in Q_\reset})\trule{a}s'_{[c]}$
\end{center}
where $c=(c_{s'}{+}j)\bmod\per_{s'}$.
\item
For each $s_{[c]}\in Q_\fixres$ we have
$s_{[c]}\trule{\varepsilon}q(0)$ where  
$q=\goto_s(c)$. 
\end{enumerate}
\end{itemize}
A \emph{configuration} $C$ is a state in $\LTSext(\A)$. 
If $C\in Q_\modp$ or $C=p(m)$
where $p\in Q_\stable$ then $C$ is \emph{stable},
otherwise $C$ is \emph{unstable}.

Moreover, we define the mapping 
\begin{center}
$\modstate:((Q_\stable\cup Q_\reset)\times\N)\rightarrow (Q_\modp\cup
Q_\fixres)$:
\end{center}
\begin{itemize}
\item
if $p\in Q_\stable$ then $\modstate(p(m))=(p,(c_s)_{s\in Q_\reset})\in
Q_\modp$ 
where 
$c_s=m\bmod \per_s$ for all $s\in Q_\reset$;
\item
if  $s\in Q_\reset$ then 
$\modstate(s(m))=s_{[c]}\in Q_\fixres$ 
where 
$c=m\bmod \per_s$.
\end{itemize}
\end{definition}

\noindent
We note that the cardinality of $Q_\modp$ might be exponential in $k$ 
(i.e. in the number of control states of $\A$).
On the other hand, $Q_\fixres$ is small;
this is a crucial fact for some claims in the next auxiliary
propositions.
We stipulate $\min\emptyset=\omega$, and recall that
$z+\omega=\omega$ for any $z\in \N$.

\begin{proposition}\label{prop:auxpmmodpm}
\hfill 
\begin{enumerate}
\item
If $(p,(c_s)_{s\in Q_\reset})\gt{w}(q,(c'_s)_{s\in Q_\reset})$ then 
for each $(p,(d_s)_{s\in Q_\reset})\in Q_\modp$ we have
$(p,(d_s)_{s\in Q_\reset})\gt{w}(q,(d'_s)_{s\in Q_\reset})$ 
where $d'_s{-}c'_s \equiv d_s{-}c_s \,(\bmod \ \per_s)$ for all $s\in
Q_\reset$.
\item
If $(p,(c_s)_{s\in Q_\reset})\gt{w}s'_{[c]}$ then 
for each $(p,(d_s)_{s\in Q_\reset}) \in Q_\modp$ we have
$(p,(d_s)_{s\in Q_\reset})\gt{w}s'_{[d]}$ where 
$d{-}c \equiv d_{s'}{-}c_{s'} \,(\bmod \,\per_{s'})$. 
\item
For any $s\in Q_\reset$ we have
$s(m)\sim \modstate(s(m))$.
\item
If $p\in Q_\stable$ then $\lev(p(m),\modstate(p(m)))=$
\\
$=\min\{z+\lev(q(0),\modstate(q(0)))\mid q\in Q_\stable$ and 
$z$ is the length of a positive path from $p(m)$ to $q(0)\}$.
\item
For any $p\in Q_\stable$, $m\in\Nat$, and $w\in\Sigma^*$
there is some small positive $d\in\N$ such that 
\begin{itemize}
\item
either 
for each
$m'$ such that $m'\equiv m\,(\bmod\, d)$
we have that $\modstate(p(m'))$ enables $w$,
\item
or 
for each
$m'$ such that $m'\equiv m\,(\bmod\, d)$
we have that $\modstate(p(m'))$ does not enable $w$.
\end{itemize}

\end{enumerate}
\end{proposition}

\begin{proof}
Points 1 and 2 can be easily shown by induction on $|w|$, using
Def.~\ref{def:extTA}.

Point 3 is obvious since $s(m)\trule{\varepsilon}q(0)$ 
and $\modstate(s(m))\trule{\varepsilon}q(0)$ for the appropriate $q\in
 Q_\stable$.

\medskip

Point 4:

We first note that 
if $p(m)\gt{w}q(n)$ is a positive path (in $\T(\A)$) 
then we have 
$\modstate(p(m))\gt{w}\modstate(q(n))$ (as can be easily
shown by induction on $|w|$).

One part of the equality, namely 
$\lev(p(m),\modstate(p(m)))\leq
\min\{\dots\}$, is thus clear; it remains to show 
\begin{equation}\label{eq:bigger}
\lev(p(m),\modstate(p(m)))\geq\min\{z+\lev(q(0),\modstate(q(0)))\mid \dots\}.
\end{equation}
The case where
$p(m)\sim\modstate(p(m))$ 
is trivial.
We
thus further consider only the cases $p(m)\not\sim\modstate(p(m))$,
and
we proceed by induction on  
$\lev(p(m),\modstate(p(m)))$.
If $\lev(p(m),\modstate(p(m)))=0$ then we obviously must have 
$m=0$, and~(\ref{eq:bigger}) is trivial in any case with $m=0$.  

Let us now assume
$m>0$, and
let $av$ ($a\in\Sigma$) be a shortest witness for
$(p(m),\modstate(p(m)))$.
We must have some $(p,a,1,q,j)\in\delta$, and thus
$p(m)\trule{a}q(m{+}j)$ and $\modstate(p(m))\trule{a}\modstate(q(m{+}j))$
(as can be easily checked). Point 3 excludes the case 
$q\in Q_\reset$, hence $q\in Q_\stable$.
By recalling Observation~\ref{obs:eqleveldrop}(2), and using the
induction hypothesis for $q(m+j),\modstate(q(m+j))$, we finish the
proof easily: $\lev(p(m),\modstate(p(m)))=1+
\lev(q(m{+}j),\modstate(q(m{+}j)))
\geq 1+z+\lev(q'(0),\modstate(q'(0)))$
where $z$ is the length of some positive path from $q(m{+}j)$ to
$q'(0)$, and $1{+}z$ is thus 
the length of some positive path from $p(m)$ to
$q'(0)$.

\medskip

Point 5:

By recalling  Points 1 and 2, we easily note the following fact:
\\
If $\modstate(p(m_1))\gt{u}C\in Q_\modp\cup Q_\fixres$
then for any $m_2$ there is $C'\in Q_\modp\cup Q_\fixres$
such that 
$\modstate(p(m_2))\gt{u}C'$;
moreover, if $\modstate(p(m_1))\gt{u}s_{[c]}$ then 
$\modstate(p(m_2))\gt{u}s_{[c]}$ for any $m_2$ such that 
 $m_2\equiv m_1\,(\bmod\,\per_s)$.

Hence if $w=uv$ where $\modstate(p(m))\gt{u}s_{[c]}$ then 
the claim is satisfied by $d=\per_s$, and otherwise it is satisfied
even by $d=1$.
\end{proof}

We recall that $C\depictlev{e}C'$ means 
 $\lev(C,C')=e$.

\begin{proposition}\label{prop:ideafewc}
\hfill
\begin{enumerate}
\item
For any $p,q\in Q_\stable$ and $m,n\in \N$
there are small positive $d_1,d_2\in\N$ such that 
for any $m',n'\in\N$ we have:
if $m'\equiv m\,(\bmod\, d_1)$ and $n'\equiv n\,(\bmod\, d_2)$ then
\begin{center}
$\lev(\modstate(p(m')),\modstate(q(n')))\leq
\lev(\modstate(p(m)),\modstate(q(n)))$.
\end{center}
\item
The set $\{\,e\mid \text{ there are }  C,C'\in Q_\modp\text{ s.t. }
C\depictlev{e}C'\}$ is small.
\end{enumerate}
\end{proposition}

\begin{proof}
Point 1:
\\
If $\modstate(p(m))\sim \modstate(q(n))$ 
then the claim is trivial. We thus assume
$\modstate(p(m))\not\sim \modstate(q(n))$ 
and let 
$w$
be a shortest witness for $(\modstate(p(m)),\modstate(q(n)))$.

By Prop.~\ref{prop:auxpmmodpm}(5), 
$p,m,w$ give rise to $d_1$ and $q,n,w$ give rise to $d_2$ such that 
precisely one of 
$\modstate(p(m')),\modstate(q(n'))$ enables $w$ when
$m'\equiv m\,(\bmod\, d_1)$
and $n'\equiv n\,(\bmod\, d_2)$.
In this case $w$ is a witness (not necessarily a shortest) for 
$(\modstate(p(m')),\modstate(q(n')))$, and the claim thus follows.

\medskip

Point 2:
\\
It is obvious that the set in Point 2 is equal to 
\begin{center}
$\{\,e\mid \modstate(p(m))\depictlev{e}\modstate(q(n))$ for some 
$p,q\in Q_\stable, m,n\in\N\}$.
\end{center}
With every tuple $(p,m,q,n)$ we associate a fixed tuple
$(d_1,d_2,c_1,c_2)$ where $d_1,d_2$ are those guaranteed by Point~1,
and $c_1=m\bmod d_1$, $c_2=n\bmod d_2$.
If two tuples  $(p,m_1,q,n_1)$, $(p,m_2,q,n_2)$   have the same 
associated tuple $(d_1,d_2,c_1,c_2)$ then 
$\lev(\modstate(p(m_1)),\modstate(q(n_1)))=
\lev(\modstate(p(m_2)),\modstate(q(n_2)))$, as follows by applying
Point~1 in both directions. 
Since the number of possible tuples $(p,q,d_1,d_2,c_1,c_2)$ is small,
we are done.
\end{proof}
The next proposition can be proved 
analogously as the previous one.

\begin{proposition}\label{prop:ideafewcfix}
\hfill
\begin{enumerate}
\item
For any $p,q\in Q_\stable$ and $m,n\in \N$
there is some small positive $d\in\N$ such that 
for any $m'\in\N$ we have:
if $m'\equiv m\,(\bmod\, d)$ then
\begin{center}
$\lev(\modstate(p(m')),q(n))\leq
\lev(\modstate(p(m)),q(n))$.
\end{center}
\item
For any (fixed) $q(n)$, the set 
$\{\,e\mid \text{ there is }  C\in Q_\modp\text{ s.t. }
C\depictlev{e}q(n)\}$ is small.
\end{enumerate}
\end{proposition}

\subsection{Eqlevels of pairs of zero 
configurations}\label{sub:elzc}

Let us recall $\LTSext(\A)$ defined in Def.~\ref{def:extTA}.
We could view the elements of $Q_\modp\cup Q_\fixres$ as additional
control states of
$\A$; in these states the counter value would play no role
and could be formally viewed as zero. This observation justifies the
name ``zero configurations'' in the following definition.

\begin{definition}\label{def:elzc}
Given a doca $\A$ as in~(\ref{eq:doca}), 
with the associated $\LTSext(\A)$ by Def.~\ref{def:extTA},
a state $C$ in $\LTSext(\A)$ is 
a \emph{zero configuration} if either 
$C\in Q_\modp\cup Q_\fixres$
or $C=p(0)$ 
where $p\in Q_\stable\cup Q_\reset$.

We define the set $\elzc\subseteq\N$ (Zero configurations Eqlevels) as follows:
\begin{center}
$\elzc=\{\,e\in\N\mid$
there are two stable zero configurations $C,C'$ s.t. 
$C\depictlev{e}C'\}$.
\end{center}
\end{definition}

\noindent
We thus have 
$\elzc=\textsc{E}_1\cup \textsc{E}_2\cup \textsc{E}_3$
where

\smallskip

$\textsc{E}_1=\{\,e\in\N\mid p(0)\depictlev{e}q(0)$ for some
$p,q\in Q_\stable\}$,

$\textsc{E}_2=
\{\,e\in\N\mid p(0)\depictlev{e}C$ for some
$p\in Q_\stable, C\in Q_\modp\}$,

$\textsc{E}_3=
\{\,e\in\N\mid C\depictlev{e}C'$ for some
$C,C'\in Q_\modp\}$.

\smallskip

\noindent
Since the set $\{\,p(0)\mid p\in Q_\stable\}$ is obviously small,
by Prop.~\ref{prop:ideafewc}(2) and~\ref{prop:ideafewcfix}(2)
we easily derive the following claim.

\begin{lemma}\label{prop:fewc}
The set $\elzc$ is small.
\end{lemma}

\noindent
The lemma does not claim that the elements of $\elzc$ are small
numbers. This will be shown
in the following  
subsections; i.e., we will prove
the next theorem
which strengthens Theorem~\ref{th:polywitness}.

\begin{theorem}\label{th:strongpolywitness}
There is a polynomial $\poly:\N\rightarrow\N$ 
such that 
$\max\,\{\,e\mid e\in\elzc\}\leq \poly(k)$ 
(for any doca $\A$ with $k$ control states).
\end{theorem}

\noindent
Let $e_0<e_1<e_2<\dots<e_f$ be the ordered
elements of $\elzc$. We have shown that $f$ is small 
but we have not yet shown that all $e_i$ are small
numbers.
 W.l.o.g. we
can assume $e_0=0$ (by adding two special control states, say). 
For proving Theorem~\ref{th:strongpolywitness}
it thus suffices to show that the
``gaps'' between $e_i$ and $e_{i+1}$, i.e. the differences
$e_{i+1}{-}e_i$, are small. 
We will later contradict
the existence of
a large gap  between $e_i=e_D$ (Down)
and $e_{i+1}=e_U$ (Up) depicted in
Figure~\ref{fig:gapineqlevels}.

\begin{figure}[ht]
\begin{center}
$e_0--e_1-\dots-e_D-----------------e_U-\dots-e_f$
\end{center}
\caption{Assumption of a large gap in $\elzc$ (to be contradicted
later)}\label{fig:gapineqlevels}
\end{figure}

But we first explore some further 
notions related to a given doca
$\A$ and the det-LTS $\LTSext(\A)$.

\subsection{Independence level}\label{sub:il}

We assume a doca $\A$ as in~(\ref{eq:doca}),
and explore a notion which we have already touched on implicitly.
\begin{definition}
For $p\in Q_\stable$, $m\in\N$ we put
\begin{center}
$\il(p(m))=\lev(p(m),\modstate(p(m)))$.
\end{center}
\end{definition}

\noindent
$\il(p(m))$ can be understood
as an ``Independence Level'' of $p(m)$
w.r.t. the 
concrete 
value $m$.

\begin{proposition}\label{prop:indlevel}
For each $p(m)$ with $\il(p(m))<\omega$ there are  
small rational numbers $\rho$, $\sigma$
(of the type $\frac{a}{b}$, $-\frac{a}{b}$ where $a,b\in\N$ are small)
and some $q\in Q_\stable$ 
such that  
\begin{center}
$\il(p(m))=\rho\cdot m +\sigma+\il(q(0))$.
\end{center}
Moreover, we can require
$\rho\geq 0$, $\rho\cdot m+\sigma\geq 0$, and if $m$ is
larger than a small bound then $\rho>0$.
\end{proposition}

\noindent
\emph{Convention.}
We will further 
assume that each $p(m)$ with  $\il(p(m))<\omega$ has a fixed
associated equality $\il(p(m))=\rho\cdot m +\sigma + e$ 
where $e=\il(q(0))\in\elzc$ and $\rho,\sigma,q$
have the claimed properties.

\begin{proof}
Suppose $\il(p(m))<\omega$. If $m=0$ then we can take $\rho=\sigma=0$
and $q=p$.
If $m>0$ 
then Prop.~\ref{prop:auxpmmodpm}(4) 
implies that there is some $q\in Q_\stable$
such that 
$\il(p(m))=|w|+\il(q(0))$ 
where $p(m)\trans{w}q(0)$ is a shortest
positive path  from $p(m)$ to $q(0)$.
(Recall the path from $C_2$ to $C_3$ in Fig.~\ref{fig:valiant} as an example.)
By Lemma~\ref{lem:valiantpositpath} we
can assume that $w$ is in the form $u_1v^iu_2$, for a short
prefix $u_1$, a short repeated cycle $v$, and a short suffix $u_2$. 
Hence $|w|=i\cdot |v|+|u_1|+|u_2|$, and
$m=i\cdot (-d) -d_1-d_2$ where $d,d_1,d_2$ are the effects of
(i.e. the counter changes caused by) $v,u_1,u_2$, respectively.
We note that  $d_2\leq 0$ and that we can assume
$d<0$.
Since  $i=\frac{m+d_1+d_2}{-d}$, we get
 $|w|=\frac{|v|}{-d}\cdot m+\frac{|v|\cdot(d_1+d_2)}{-d} +
 |u_1|+|u_2|$.
As $\il(p(m))=|w|+\il(q(0))$, all the claims follow easily.
\end{proof}

Figure~\ref{fig:ilpm} depicts $\il(p(m))$
for a fixed $p\in Q_\stable$ and for a few values $m$, by using 
black circles  $\bullet$; $e_1, e_2, e_3$ are elements of $\elzc$
corresponding to $\il(q(0))$ for several $q$.
There might be some ``irregular'' values $\il(p(m))=z+\il(q(0))$ for
small $m$ and small $z$
but for $m$
larger than a small bound the values $\il(p(m))$ lie on few
lines, starting near some $e_j$ and having small slopes. 
(In fact, we have $1\leq\frac{|v|}{|\effect(v)|}\leq k$ for 
the respective cycles $v$ in $w=u_1v^iu_2$; 
the unit-length for the vertical axis is thus smaller than 
for the horizontal axis in Fig.~\ref{fig:ilpm}.)
The circles $\bullet$ and $\circ$ on one depicted line can correspond 
to the pairs 
$(m_0,z_0+\il(q(0)))$, $(m_0+d, z_0+d'+\il(q(0)))$,
$(m_0+2d, z_0+2d'+\il(q(0)))$,
$\dots$ where $d=|\effect(v)|$ and $d'=|v|$
(and 
$z_0=|u_1u_2|$, $z_0+d'=|u_1vu_2|$, 
 $z_0+2d'=|u_1v^2u_2|$, 
$\dots$).
A white circle $\circ$ depicts that the respective value,
corresponding 
to a positive path $p(m_0{+}id)\gt{u_1v^iu_2}q(0)$, 
is not $\il(p(m_0{+}id))$ since there is another, and shorter, witness 
in this case.

\begin{figure}[ht]
\centering
\epsfig{scale=0.2,file=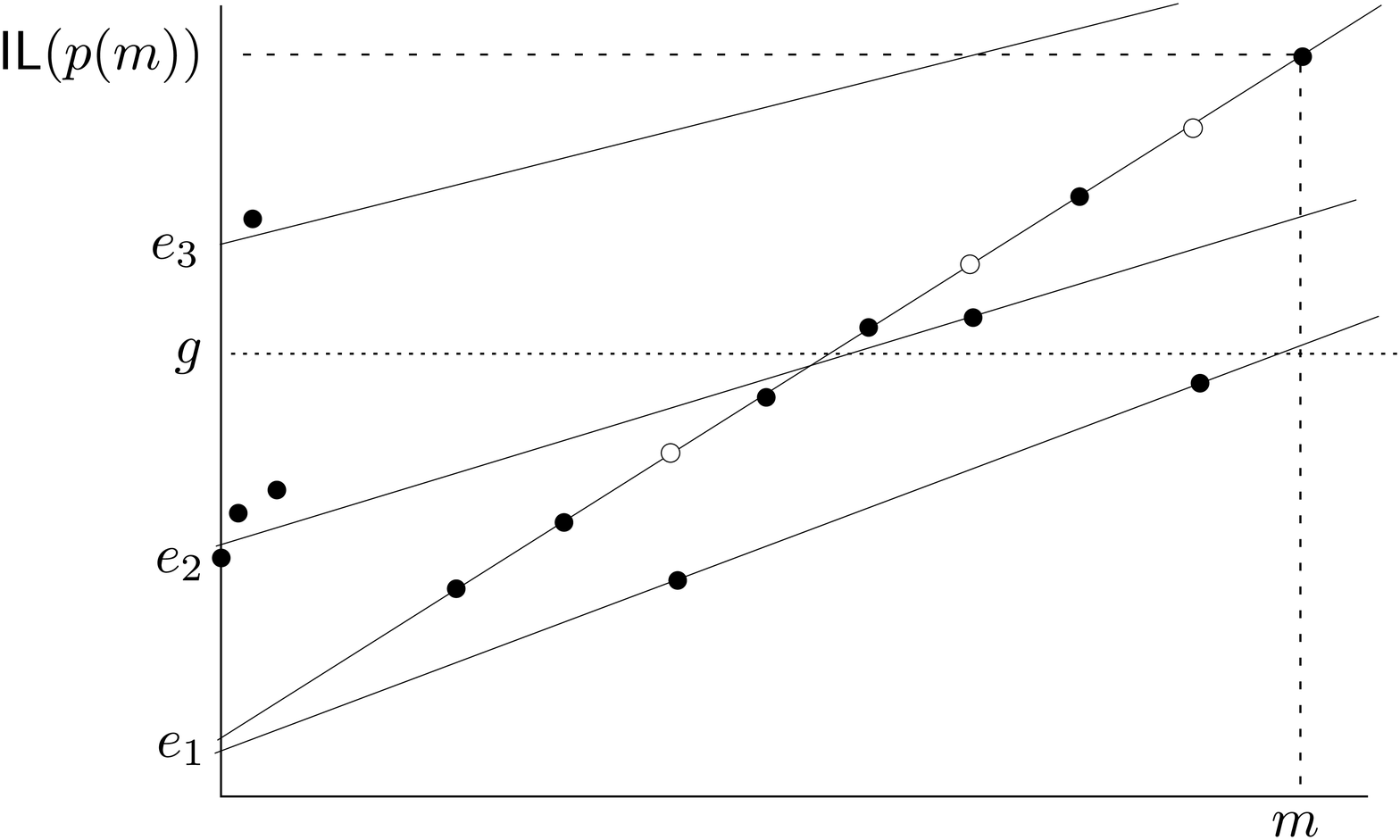}
\caption{Illustrating $\il(p(m))$ as a function of $m$}\label{fig:ilpm}
\end{figure}

We now observe some further facts for later use.

\begin{proposition}\label{prop:conseqilpm}
\hfill
\begin{enumerate}
\item
For each $g\in\N$ there are only few $p(m)$ such that 
$\il(p(m))=g$. 
\item
For any $p(m)$ where $\il(p(m))<\omega$ there are some small
numbers $\base\geq 0$ and $\per>0$ such that the following condition
holds:
\\
for any $m'$ such that $\base\leq m'<m$ and
$m'\equiv m\, (\bmod\ \per)$  
we have $\il(p(m'))<\il(p(m))$.
\end{enumerate}
\end{proposition}
\begin{proof}
Point 1 is intuitively clear from the horizontal line at level $g$ 
in Fig.~\ref{fig:ilpm}.
Formally,  we look
when we can have $g=\rho\cdot m + \sigma +e$ where 
$\il(p(m))=\rho\cdot m + \sigma +e$ is the equality associated with
some $p(m)$ (by Convention after Prop.~\ref{prop:indlevel}).
Since there are only few possibilities for $\rho,\sigma,e$, and we
can have $\rho=0$ only for few (small) values $m$, there are only few
possible $m$ which might fit. 

Point 2 has been intuitively shown by the line with 
black and white circles in Fig.~\ref{fig:ilpm} and by the respective
discussion. To be more formal, we recall 
that $\il(p(m))=|w|+\il(q(0))$ for some $q\in Q_\stable$
and a shortest positive path $p(m)\gt{w}q(0)$ from $p(m)$ to $q(0)$.
We assume
$w=u_1v^iu_2$ for short $u_1,v,u_2$
as in 
the proof of Prop.~\ref{prop:indlevel}.
Hence if $m$ is bigger than a small bound then $i>0$.
Let $\per=|\effect(v)|$.
Then we have
\begin{quote}
$p(m-\per)\xrightarrow{u_1v^{i-1}u_2}q(0)$,

$p(m-2\cdot\per)\xrightarrow{u_1v^{i-2}u_2}q(0)$,

$\dots$,

$p(m-x\cdot\per)\xrightarrow{u_1v^{i-x}u_2}q(0)$
\end{quote}
for $x=(m{-}\base)\div\per$ where we put
$\base=  |u_1|{+}|u_2|{+}|v|$ to be safe, i.e. to guarantee that
$u_1v^{i-j}u_2$ is indeed enabled in $p(m-j\cdot\per)$,
for all $j=1,2,\dots,x$.
Since $p(m-j\cdot\per)\xrightarrow{u_1v^{i-j}u_2} q(0)$ 
is a positive path,
we have $\il(p(m-j\cdot\per))\leq
|u_1v^{i-j}u_2|+\il(q(0))$ by Prop.~\ref{prop:auxpmmodpm}(4).
Since 
$|u_1v^{i-j}u_2|+\il(q(0))<|w|+\il(q(0))=\il(p(m)$, we are done.
\end{proof}

\subsection{Eqlevel tuples}\label{sub:eqltuples}

We introduce the eqlevel tuples illustrated in
Fig.~\ref{fig:quadruple}, assuming a given doca $\A$
as in~(\ref{eq:doca}), with the associated det-LTSs $\T(\A)$ and $\LTSext(\A)$.
A simple property of these tuples considerably simplifies the later
analysis. 

\begin{figure}[ht]
\centering
\epsfig{scale=0.4,file=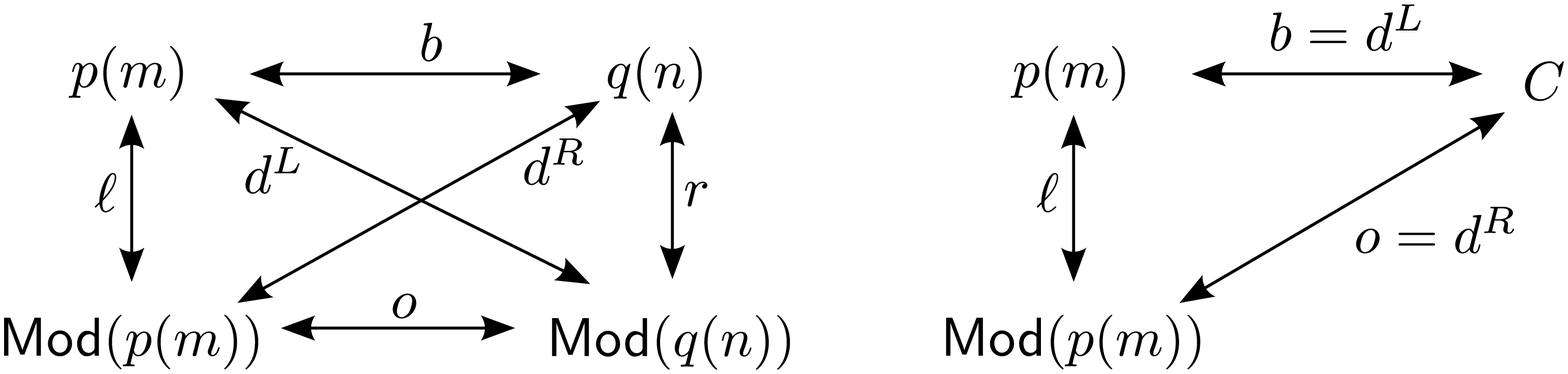}
\caption{
Eqlevel tuple $(b,\ell,r,o,d^L,d^R)$
associated to a pair $(p(m),q(n))$,
and to a pair $(p(m),C)$ where $C\in Q_\modp$}\label{fig:quadruple}
\end{figure}

\begin{definition}\label{def:quadruple}
Each pair $(p(m),q(n))$
of stable configurations in $\T(\A)$
has the \emph{associated eqlevel tuple}
$(b,\ell,r,o,d^L,d^R)$ (of elements from $\N\cup\{\omega\}$) defined
as follows:
\begin{itemize}
\item $b=\lev(p(m), q(n))$ (Basic),
\item
$\ell=\il(p(m))$ (Left), 
\item
$r=\il(q(n))$ (Right),
\item
$o=\lev(\modstate(p(m)), \modstate(q(n))$ (mOd),
\item
$d^L=\lev(p(m), \modstate(q(n))$ (Diagonal Left),
\item
$d^R=\lev(q(n), \modstate(p(m))$ (Diagonal Right).
\end{itemize}
Each pair $(p(m),C)$ where $C\in Q_\modp$ and 
$p(m)$ is a stable configuration in
$\T(\A)$ has    the \emph{associated eqlevel tuple}
$(b,\ell,r,o,d^L,d^R)$ defined as follows:
\begin{itemize}
\item  $b=d^L=\lev(p(m), C)$,
\item $\ell=\il(p(m))$, 
\item $r=\omega$, 
\item $o=d^R=\lev(\modstate(p(m)), C)$.
\end{itemize}
\end{definition}

\medskip
We could similarly associate a tuple to $(C,q(n))$ but this is not
needed in later reasoning.
The following trivial fact 
yields an important corollary for the eqlevel tuples; it holds for
general LTSs but we confine ourselves to the introduced det-LTSs.

\begin{proposition}\label{prop:mineqlevel}
Given states $s_1,s_2,\dots,s_m$
in a  det-LTS where  $m\geq 2$
and $s_1\depictlev{e_1}s_2$, $s_2\depictlev{e_2}s_3$, $\dots$,
$s_{m-1}\depictlev{e_{m-1}}s_m$, $s_m\depictlev{e_m}s_1$, 
the minimum of $\{e_1,e_2,\dots,e_m\}$ cannot be $e_i$ for just
one $i$.
\end{proposition}

\begin{proof}
We assume by contradiction that $\min\{e_1,\ldots,e_m\}=e_i$ for just one
$i\in\{1,\ldots,m\}$; w.l.o.g. we assume
$i=1$, and we note that $e_1<\omega$ (since $m\geq 2$).
Then we have $s_2\sim_{e_1+1}s_3\sim_{e_1+1}s_4\,\cdots\sim_{e_1+1} s_m\sim_{e_1+1} s_1$
and hence $s_1\sim_{e_1+1}s_2$ by transitivity and symmetry of
$\sim_{e_1+1}$; this contradicts the assumption
$s_1\depictlev{e_1}s_2$.
\end{proof}

\begin{corollary}\label{prop:cormineqlevel}
In the ``triangle'' $(b,\ell,d^R)$, 
we always have $b=\ell$ or $b=d^R$ or $\ell=d^R$ (or $b=\ell=d^R$) as
the minimum. Similarly for the ``triangles'' $(d^R,r,o)$,
$(b,d^L,r)$, and $(\ell,d^L,o)$. In the ``rectangle''
$(b,\ell,r,o)$, the minimum is also achieved by at least two elements
(concretely by $b=\ell$, $b=r$, $b=o$, $\ell=r$, $\ell=o$, or $r=o$).
\end{corollary}

\subsection{Paths in {\large \bf $\T(\A)\times\T(\A)$}}\label{sub:paths in tt}

Since we are interested in comparing two states in a det-LTS $\T$,
it is useful to define the product $\T\times \T$; 
the transitions in $\T\times\T$ are just the 
letter-synchronized pairs of transitions in $\T$.
Eqlevel-decreasing
paths in $\T\times \T$ will be of particular interest.
A formal definition follows.

\begin{definition}
Let
$\T=(S_\stable, S_\unstable,\Sigma,
(\trule{a})_{a\in\Sigma},\trule{\varepsilon})$ be a det-LTS.
We define the det-LTS 
\begin{center}
$\T\times\T=(S_\stable\times S_\stable,
S'_\unstable, 
\Sigma,
(\trule{a})_{a\in\Sigma},\trule{\varepsilon})$
\end{center}
where 
$S'_\unstable=(S_\stable\times S_\unstable)
\cup (S_\unstable\times S_\stable)
\cup (S_\unstable\times S_\unstable)$
and the transitions are defined as follows:
\begin{enumerate}
\item
If $s,t\in S_\stable$ and 
$s\trule{a}s'$ and $t\trule{a}t'$ (for $a\in\Sigma$) then
$(s,t)\trule{a}(s',t')$.
\item
If $s\in S_\stable$, $t\in S_\unstable$,
and $t\trule{\varepsilon}t'$ then  
$(s,t)\trule{\varepsilon}(s,t')$.
\item
If $s\in S_\unstable$, $t\in S_\stable$,
and $s\trule{\varepsilon}s'$ then  
$(s,t)\trule{\varepsilon}(s',t)$.
\item
If 
$s\trule{\varepsilon}s'$ and 
$t\trule{\varepsilon}t'$
then  
$(s,t)\trule{\varepsilon}(s',t')$.
\end{enumerate}

\medskip

\noindent
A \emph{path}
$(s_0,s'_0)\gt{a_1}(s_1,s'_1)\gt{a_2}(s_2,s'_2)$ $\,\cdots$ $\gt{a_z} (s_z,s'_z)$
in $\T\times\T$  (where $(s_i,s'_i)\in S_\stable\times S_\stable$ by
Def.~\ref{def:paths})
is \emph{eqlevel-decreasing}
if $\lev(s_i,s'_i)>\lev(s_{i+1},s'_{i+1})$ for all
$i\in\{0,1,,\dots,z{-}1\}$.
\end{definition}

\medskip

We can easily verify that $\T\times\T$ is indeed a det-LTS.
We also note that in eqlevel-decreasing paths we must have
$\lev(s_{i+1},s'_{i+1})=\lev(s_i,s'_i)-1$, by
Observation~\ref{obs:eqleveldrop}.
We also observe:

\begin{observation}\label{prop:shortestsubpath}
\hfill
\begin{enumerate}
\item
Any subpath of an eqlevel-decreasing path
in $\T\times \T$ is a shortest
path from its start to its end.
\item
Suppose the path $(s,t)\trans{w}(s',t')$
is eqlevel-decreasing.
If $(s,t)\trans{v}(s'',t'')$ where $|v|<|w|$ then
$\lev(s'',t'')>\lev(s',t')$.
\end{enumerate}
\end{observation}

We now look at $\T(\A)\times\T(\A)$ for a doca $\A$.
\begin{definition}
We call
$(p(m),q(n))\gt{a}(p'(m'),q'(n'))$ 
a  \emph{reset step} (in $\T(\A)\times\T(\A)$) if
at least one of component-steps $p(m)\gt{a}p'(m')$,
 $q(n)\gt{a}q'(n')$ is  a reset step in $\T(\A)$.
If precisely one of component-steps is a reset step then 
$(p(m),q(n))\gt{a}(p'(m'),q'(n'))$ is 
a \emph{one-side reset step},
if both component-steps are reset steps then 
$(p(m),q(n))\gt{a}(p'(m'),q'(n'))$ is 
a \emph{both-side reset step}.
\end{definition}

We note that one of $m',n'$ is $0$ when $(p(m),q(n))\gt{a}(p'(m'),q'(n'))$
is a one-side reset step, and  $m'=n'=0$ when it is a both-side reset
step.

\begin{figure}[ht]
\centering
\epsfig{scale=0.17,file=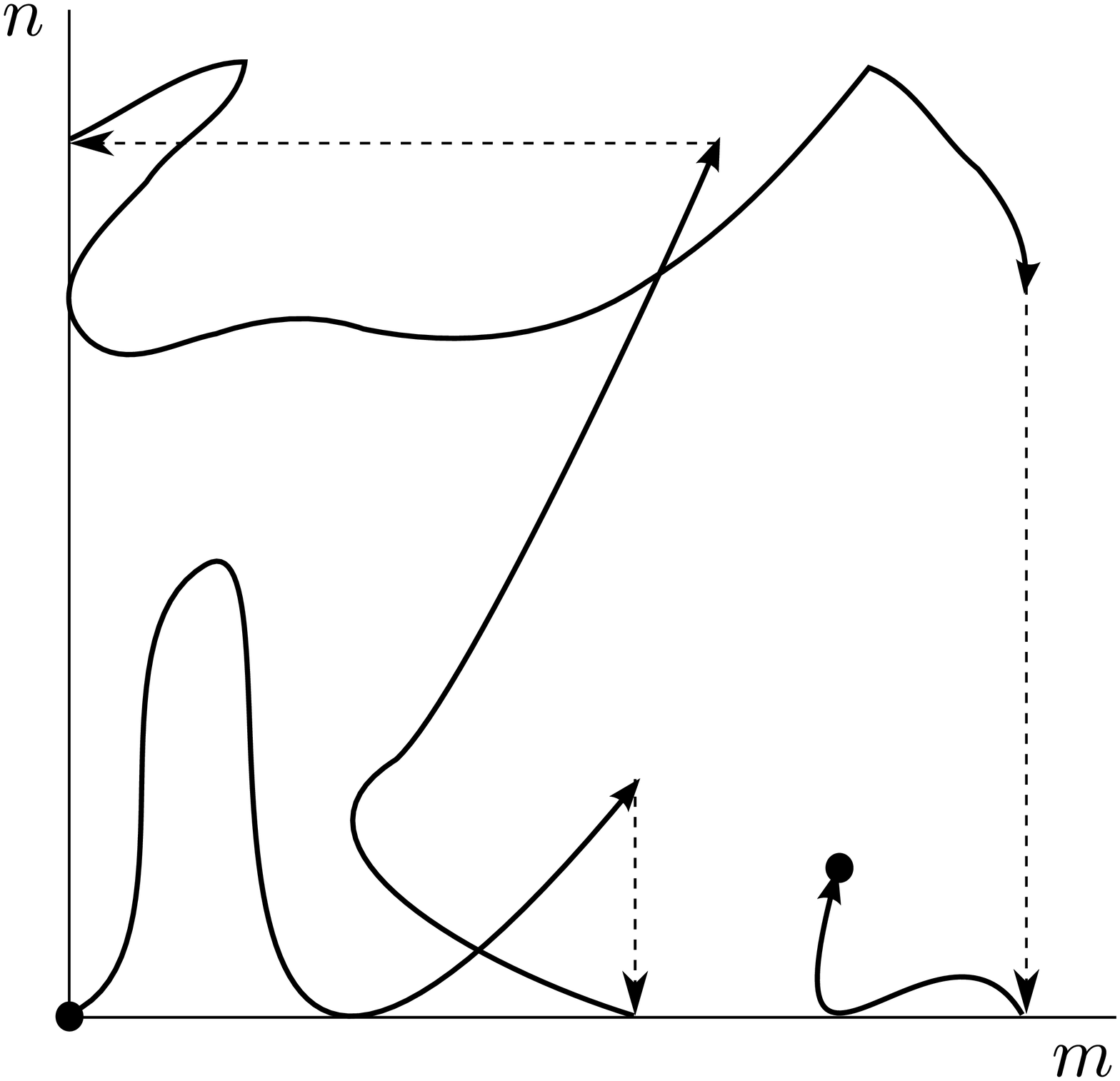}
\caption{A path from $(p(0),q(0))$
in $\T(\A)\times\T(\A)$
(with some one-side resets),
projected to $\N\times\N$.}\label{fig:examplepathinTATA}
\end{figure}

Fig.~\ref{fig:examplepathinTATA} shows an example
of a path $\T(\A)\times\T(\A)$,
projected to  $\N\times\N$ 
(a pair $(p(m),q(n))$ is projected to $(m,n)$);
the dotted lines
represent one-side reset steps.
Theorem~\ref{th:polywitness} claims, in fact,
that the eqlevel-decreasing paths in $\T(\A)\times\T(\A)$
which start from pairs of zero configurations are short.

\subsection{$\il$-equality lines}\label{sub:eq-lines}

We assume a fixed doca $\A$, and
consider the cases $\il(p(m))=\il(q(n))<\omega$
(i.e., $\ell=r<\omega$ in Fig.~\ref{fig:quadruple});
we explore  
what we can say about the respective points $(m,n)\in\N\times\N$.
By Convention after Prop.~\ref{prop:indlevel}, each such case
has the associated equalities
$\il(p(m))=\rho\cdot m+\sigma+e$ and 
$\il(q(n))=\rho'\cdot n+\sigma'+e'$, and 
 $\il(p(m))=\il(q(n))$ thus implies
$\rho\cdot m+\sigma+e=\rho'\cdot n + \sigma' + e'$.

Only in few cases we have $\rho=0$ or $\rho'=0$
(which is clear by Prop.~\ref{prop:indlevel} and
Prop.~\ref{prop:conseqilpm}(1));
in the other (many) cases
we have 
$n=\frac{\rho}{\rho'}m+\frac{(\sigma{-}\sigma')+(e-e')}{\rho'}$ 
where $\frac{\rho}{\rho'}>0$.
This naturally leads to the following notions
(illustrated in Fig.~\ref{fig:lines}).

\begin{figure}[ht]
\centering
\epsfig{scale=0.12,file=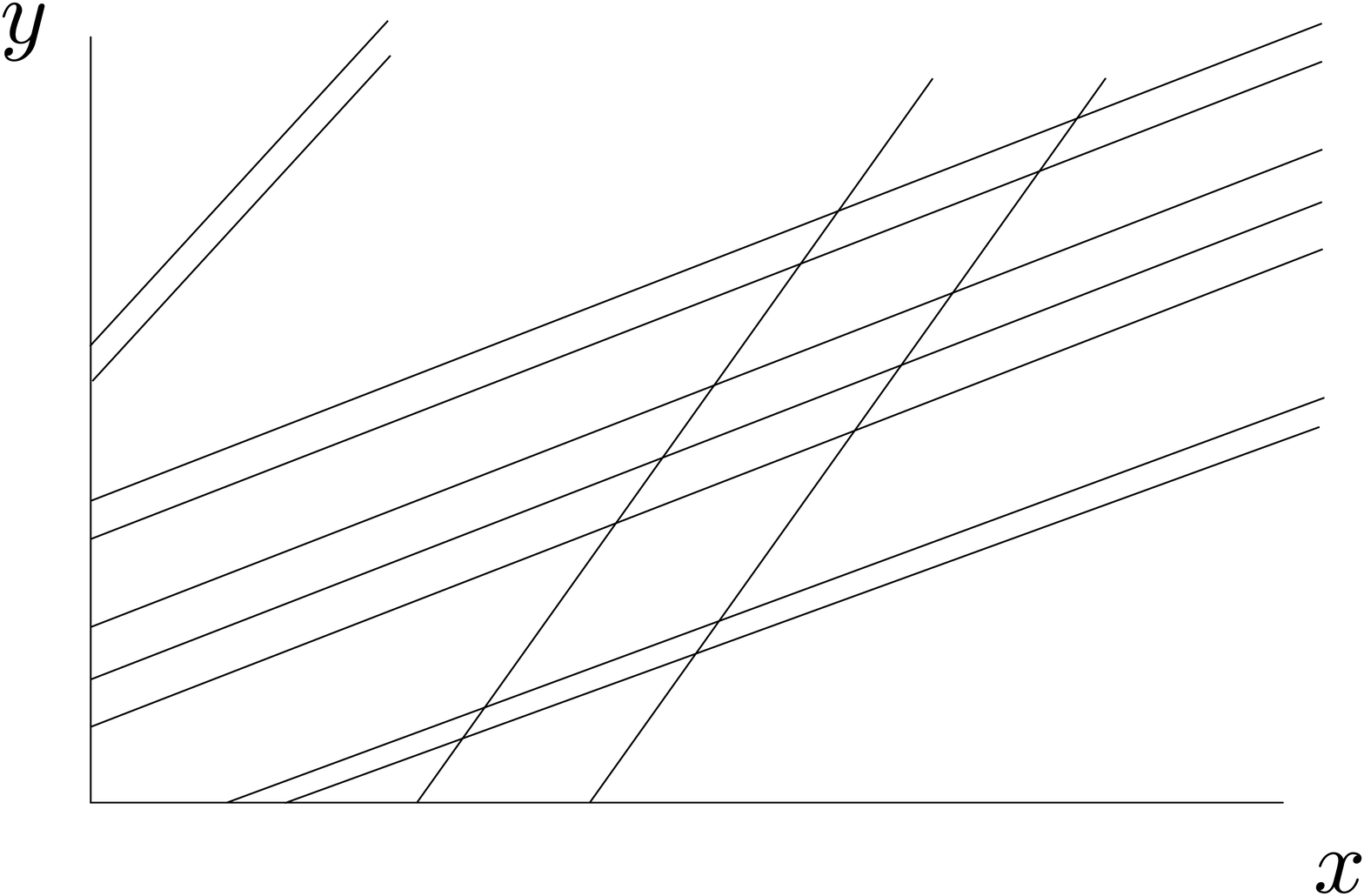}
\caption{
A sketch of $\il$-equality lines (in reality, lines contain only
points with integer coordinates)}\label{fig:lines}
\end{figure}

\begin{definition}
A pair $(\mu,\tau)$ of rational numbers is a \emph{valid
slope-shift pair} if there are some $p(m)$, $q(n)$ with the
associated equalities 
$\il(p(m))=\rho\cdot m+\sigma+e$ and 
$\il(q(n))=\rho'\cdot n+\sigma'+e'$ such that 
$\rho\cdot m+\sigma+e=\rho'\cdot n + \sigma' + e'$,
$\rho>0$, $\rho'>0$, $\mu=\frac{\rho}{\rho'}$, 
$\tau=\frac{(\sigma{-}\sigma')+(e-e')}{\rho'}$. 

Each valid slope-shift pair  $(\mu,\tau)$ defines an
\emph{$\il$-equality line}, or just a \emph{line} for short,
namely the set $\{(x,y)\in\N\times\N\mid y=\mu\cdot x +\tau\}$.

Any maximal set of parallel lines (having the same slope but
various shifts)
is a \emph{line-bunch}. (The maximality is taken w.r.t. set inclusion.)
We say that $(x,y)\in \N\times\N$ is in a line-bunch $H$ if $(x,y)$
is in a line in $H$.
\end{definition}
Though each line
contains at least one $(m,n)$
such that $\il(p(m))=\il(q(n))<\omega$ for some $p,q$, the definition
does not assume anything more specific about lines.
The line-bunches can have various ``gaps'', and if a point
$(x,y)$ is not in a line-bunch $H$ then it can still
lie between two lines
from $H$.
The following proposition is easy to verify. 

\begin{proposition}\label{prop:fewlines}
\hfill
\begin{enumerate}
\item
There are only few lines, and thus also few line-bunches. 
\\
The set $\{(x,y)\in\N\times\N\mid (x,y)\in L_1\cap
L_2$ for two different lines $L_1, L_2\}$ is small.
\item
There are only few pairs $(p(m),q(n))$ where 
$\il(p(m))=\il(q(n))<\omega$ and $(m,n)$ is not in a line.
\end{enumerate}
\end{proposition}

\subsection{Eqlevel-decreasing line-climbing 
paths are short}\label{sub:line-climbing}

We recall Fig.~\ref{fig:gapineqlevels} 
which assumes a large gap $e_U{-}e_D$; 
to finish a proof of
Theorem~\ref{th:strongpolywitness}, 
we aim to show that all gaps 
in $\elzc$ are, in fact, small.
In the next subsection (\ref{subsec:gapssmall}) we show that 
a large gap  $e_U{-}e_D$
would entail a long eqlevel-decreasing line-climbing 
path in $\T(\A)\times\T(\A)$ (depicted in Fig.~\ref{fig:cyclesinlines}).
In this subsection we show that all such paths are, in fact, short.
Fig.~\ref{fig:cyclesinlines} illustrates a line-climbing path from 
a pair projected to $P_1$ to a larger pair projected to $P_2$.
The cyclicity and further structures in the figure will be discussed later.

\begin{figure}[ht]
\centering
\epsfig{scale=0.25,file=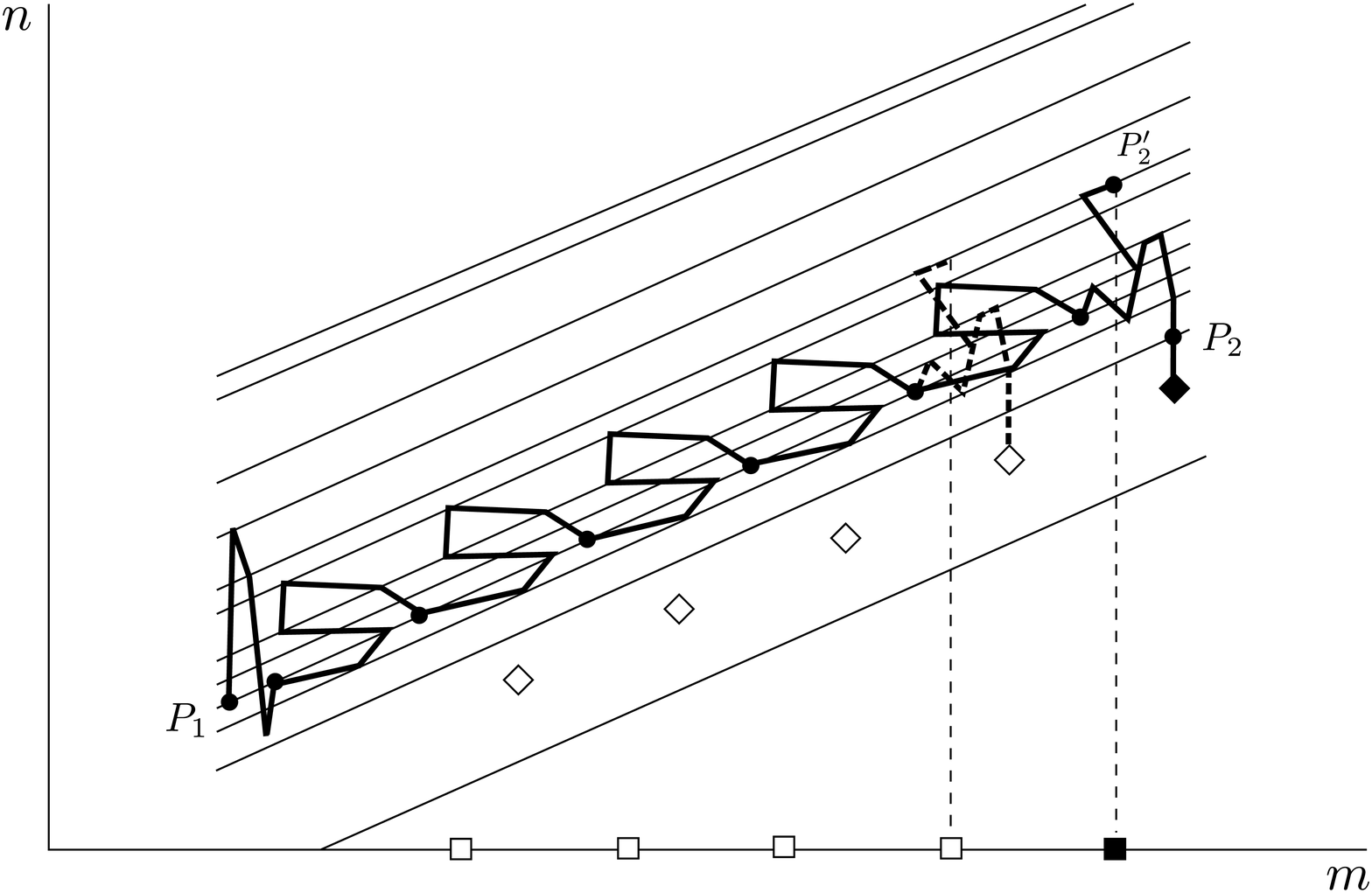}
\caption{A line-climbing path (projections of all 
visited configuration-pairs are in IL-equality lines in one line-bunch)}\label{fig:cyclesinlines}
\end{figure}

\begin{definition}\label{def:bunchclimb}
A \emph{path in} $\T(\A)\times\T(\A)$ is
\emph{positive} if each pair  $(p(m),q(n))$ in the path satisfies 
$m>0$, $n>0$; this entails that there are no reset steps in the path.

A positive path 
$(p_0(m_0),q_0(n_0))\gt{a_1}$ $(p_1(m_1),q_1(n_1))\gt{a_2}$ $\cdots$
 $\gt{a_{z}}(p_z(m_z),q_z(n_z))$
 is \emph{line-climbing} if $m_0<m_z$
 and all $(m_i,n_i)$, for $i=0,1,2,\dots,z$,
 are in one line-bunch. 
\end{definition}

We do not require that $(m_0,n_0)$ and $(m_z,n_z)$ are in the same
line, and we might have $n_z\leq n_0$;
hence ``line-climbing'' might be  understood as a shorthand for 
``(left-to-right) line-bunch climbing''.

To get some intuition for what follows, imagine that
Fig.~\ref{fig:cyclesinlines} illustrates the projection of a ``cyclic''
line-climbing eqlevel-decreasing path from $P_1$ to $P_2$ which 
is followed by a simple step leading out of the respective
line-bunch, namely to the black-diamond point.
Cutting off the copies of the cycle in the path would give rise to the
sequence of white-diamond points. 

Fig.~\ref{fig:cyclesinlines} also illustrates a similar path from
$P_1$ to $P'_2$ which is followed by another type of leaving the
line-bunch, namely by a  one-side
reset step to the black-box point.
Cutting off the copies of the cycle in the path would now give rise to the
sequence of white-box points. 

If the original path, including the line-bunch leaving step,
is eqlevel-decreasing
then the eqlevel of the ``exit pair'' 
(the black diamond or the black box)
is less than the eqlevels of all
``earlier exit pairs''
(white diamonds or white boxes)
 (recall Observation~\ref{prop:shortestsubpath}(2)).
The sequence of white-diamond (or white-box) points, finished by the 
black-diamond (or black-box) point, inspires the following definition.

\begin{definition}\label{def:strangeseq}
For $p,q\in Q_\stable$, a \emph{sequence} of pairs
\begin{center}
$(p(m_0),q(n_0)), (p(m_1),q(n_1)), (p(m_2),q(n_2)),
\dots, (p(m_z),q(n_z))$
\end{center}
where $z\geq 1$
is \emph{strange periodic}
if the following conditions hold:
\begin{enumerate}
\item
$(m_i,n_i)=(m_0+i\cdot c_1,n_0+i\cdot c_2)$ for some
$c_1,c_2\in\N$ and $i=0,1,\dots,z$;
\item
$\lev(p(m_i),q(n_i))>\lev(p(m_z),q(n_z))$
for all $i\in\{0,1,\dots,z{-}1\}$ (hence $c_1>0$ or $c_2>0$);
\item
the pairs $(m_0,n_0),(m_1,n_1),\dots, (m_z,n_z)$
are not all in one $\il$-equality line.
\end{enumerate}
\end{definition}
\noindent
Prop.~\ref{prop:fewlines} implies
that in any
strange periodic sequence there are 
only few pairs $(p(m_i),q(n_i))$ such that
$\il(p(m_i))=\il(q(n_i))<\omega$.

We now show that 
all strange periodic sequences are short, 
and then we derive that all 
line-climbing eqlevel-decreasing paths are short.
(Fig.~\ref{fig:cyclesinlines} suggests that such paths can be assumed 
to use a ``cycle'';
this will be established later by another use of
Lemma~\ref{lem:valiantpositpath}.)

\begin{proposition}\label{prop:strangeshort}
Strange periodic sequences are short.
\end{proposition}

\begin{proof}
Let us assume a strange periodic 
sequence
\begin{equation}\label{eq:strangeper}
(p(m_0),q(n_0)), (p(m_1),q(n_1)), (p(m_2),q(n_2)),\dots, (p(m_z),q(n_z))
\end{equation}
as in Def.~\ref{def:strangeseq}.
Hence there are $c_1,c_2\in\N$
such that 
$(m_i,n_i)=(m_0+i\cdot c_1,n_0+i\cdot c_2)$ for $i=0,1,\dots,z$;
moreover, $c_1>0$ or $c_2>0$, 
and the pairs in~(\ref{eq:strangeper})  are thus pairwise different.

For $i\in\{0,1,\dots,z\}$, by
\begin{center}
$(b_i,\ell_i,r_i,o_i,d^L_i,d^R_i)$
we denote the eqlevel tuple associated with
$(p(m_i),q(n_i))$
\end{center}
(recall Fig.~\ref{fig:quadruple} and
Cor.~\ref{prop:cormineqlevel}).
As we already noted,
we  have 
\begin{equation}\label{eq:fewlisr}
\ell_i=r_i<\omega \textnormal{ only for few } i\in\{0,1,2,\dots,z\}. 
\end{equation}
We now explore certain ``dense''
periodic subsequences of~(\ref{eq:strangeper}).
By a \emph{periodic subsequence}, with the \emph{period} $\per>0$ and the
\emph{base} $b\geq 0$, we mean
the sequence of pairs $(p(m_j),q(n_j))$ where 
$j$ ranges over the \emph{index set}
\begin{center}
$\J=\{z-x\cdot\per, z-(x{-}1)\cdot\per, z-(x{-}2)\cdot\per,
\dots, z-2\cdot\per, z-\per\}$
\end{center}
for $x=(z{-}b)\div\per$. 
If both $b$ and $\per$ are small
(i.e., bounded by $\poly(k)$ for a fixed polynomial $\poly$
independent of 
the assumed doca $\A$ with $k$ control states)
then we say that this periodic \emph{subsequence} is \emph{dense}. 
We note that
\begin{center}
if a dense subsequence is short then 
the whole sequence~(\ref{eq:strangeper}) is short (i.e., $z$ is small).
\end{center}
By (2) in Def.~\ref{def:strangeseq} we have $b_i>b_z$ for all $i<z$,
hence also $b_j>b_z$ for all $j\in\J$ where $\J$ is the index set 
of a periodic subsequence.
Using Prop.~\ref{prop:conseqilpm}(2),
we now observe that there is 
a dense subsequence, with the index set $\J_1$, where
$\ell_j\leq \ell_z$ for all $j\in\J_1$ (when $\ell_z<\omega$
and $c_1>0$ then we
can even establish $\ell_j<\ell_z$). 
Similarly there is a dense subsequence,
with the index set $\J_2$,
where $r_j\leq r_z$ for all
$j\in\J_2$. 
By using Prop.~\ref{prop:ideafewc}(1) we derive that there is 
also a dense subsequence, with the index set $\J_3$,
where $o_j\leq o_z$ for all $j\in\J_3$. 
(Given $d_1,d_2$ guaranteed for $p,q,m_z,n_z$ by 
Prop.~\ref{prop:ideafewc}(1), we can
take $d_1\cdot d_2$ as the period of the subsequence.)

Moreover, if $c_2=0$, and thus $q(n_i)=q(n_0)$ in all pairs 
in~(\ref{eq:strangeper}), then Prop.~\ref{prop:ideafewcfix}(1)
implies that there is a dense subsequence,
 with the index set $\J_4$,
where 
$d^R_j\leq d^R_z$ for all $j\in\J_4$. 

We now perform a case analysis.
\begin{enumerate}
\item
$c_1>0$, $c_2=0$ (the case $c_1=0$, $c_2>0$ is symmetric)

Here we have $q(n_i)=q(n_0)$ in all pairs 
in~(\ref{eq:strangeper}). Considering the triangle
$\{b_z,\ell_z,d^R_z\}$ (recall Fig.~\ref{fig:quadruple} and
Cor.~\ref{prop:cormineqlevel}), we note that 
we must have $\ell_z\leq b_z<\omega$ or 
$d^R_z\leq b_z<\omega$.
Hence there is a dense subsequence, indexed by $\J$, where 
 $\ell_j\leq \ell_z\leq b_z<b_j$ for all $j\in\J$, or 
 $d^R_j\leq d^R_z\leq b_z<b_j$ for all $j\in\J$.
In both cases,  Cor.~\ref{prop:cormineqlevel} implies
that $\ell_j=d^R_j<b_j$ for all $j\in\J$. 
Since each $d^R_j$ belongs to the set
$\{\,e\mid \modstate(p(m))\depictlev{e}q(n_0)$ for some $m\,\}$,
Prop.~\ref{prop:ideafewcfix}(2) implies that 
the set $\{d^R_j\mid j\in\J\}=\{\ell_j\mid j\in\J\}$ is small.
Prop.~\ref{prop:conseqilpm}(1) then implies
that the set $\{p(m_0+j\cdot c_1)\mid j\in\J \}$ is small;
this implies 
that $\J$ is small and thus~(\ref{eq:strangeper}) is short.

\item
$c_1>0$, $c_2>0$

Looking at the rectangle $\{b_z,\ell_z,r_z,o_z\}$, we note that 
we have $\ell_z\leq b_z<\omega$ or $r_z\leq b_z<\omega$
or $o_z\leq b_z<\omega$.
Hence there is a dense subsequence, indexed by $\J$, where 
 $\ell_j\leq \ell_z\leq b_z<b_j$ for all $j\in\J$, or 
 $r_j\leq r_z\leq b_z<b_j$ for all $j\in\J$,
or $o_j\leq o_z\leq b_z<b_j$ for all $j\in\J$.
In any case,
Cor.~\ref{prop:cormineqlevel} implies that for each $j\in\J$
we have $\ell_j=r_j<\omega$ or $\ell_j=o_j<\omega$ or $r_j=o_j<\omega$. 

We note that the set 
$\{(p(m_0+j\cdot c_1),q(n_0 + j\cdot c_2))\mid j\in\J,
\ell_j=r_j<\omega\}$
is small by~(\ref{eq:fewlisr}), 
and the set $\{(p(m_0+j\cdot c_1),q(n_0 + j\cdot c_2))\mid j\in\J,
\ell_j=o_j<\omega$ or $r_j=o_j<\omega \}$
is small by 
Prop.~\ref{prop:ideafewc}(2)
and Prop.~\ref{prop:conseqilpm}(1).
This implies that $\J$ is small and thus~(\ref{eq:strangeper}) is short. 
\end{enumerate}
\end{proof}

\begin{proposition}\label{prop:climbshort}
Eqlevel-decreasing line-climbing paths are short.
\end{proposition}
\begin{proof}
We consider an eqlevel-decreasing line-climbing path in a fixed
line-bunch $H$, in the form
\begin{equation}\label{eq:lineclimb}
(p_0(m_0),q_0(n_0))\gt{a_1}
(p_1(m_1),q_1(n_1))\gt{a_2}  
\cdots
\gt{a_z}(p_z(m_z),q_z(n_z))
\end{equation}
as in Def.~\ref{def:bunchclimb}; we recall that 
the path is positive and
$m_0<m_z$.
Moreover, we assume that~(\ref{eq:lineclimb})
can not be prolonged by one step, by which we mean that 
one of 
the following conditions holds.
\begin{enumerate}
\item
$\lev(p_z(m_z),q_z(n_z))=0$.
\item
Each eqlevel decreasing 
step $(p_z(m_z),q_z(n_z))\trans{a}(p'(m'),q'(n'))$ 
is of one of the following types:
\begin{enumerate}
\item
it is a (one-side
or both-side) reset step,
\item
it spoils the ``one line-bunch property'' 
($(m',n')$ is out of the line-bunch $H$),  
\item
$m_0\geq m'$ (which entails $m_z=m_0{+}1$ and $m'=m_0$ when the step
is simple).
\end{enumerate}
\end{enumerate}
E.g., $(p_0(m_0),q_0(n_0))$ might be projected to $P_1$ 
in Fig.~\ref{fig:cyclesinlines}; the projections $P_2$ and $P'_2$
represent two possible end-pairs $(p_z(m_z),q_z(n_z))$ after which the
line-bunch $H$ is left by eqlevel decreasing steps.

We now note that the path~(\ref{eq:lineclimb}) in $\T(\A)\times\T(\A)$
can be alternatively presented as
\begin{equation}\label{eq:lineclimblinenot}
((p_0,q_0,L_0),m_0)\gt{a_1}
((p_1,q_1,L_1),m_1)\gt{a_2}  
((p_2,q_2,L_2),m_2)\gt{a_3}  
\cdots
\gt{a_z}((p_z,q_z,L_z),m_z)
\end{equation}
where $L_i$ denotes the (unique) $\il$-equality line in 
the line-bunch $H$
which contains $(m_i,n_i)$.
This presentation looks like a path in $\T(\B)$ for a doca $\B$ 
which has the triples $(p,q,L)$
as the control states
(where $p,q$ are stable control states
of $\A$ and $L$ is a denotation of a line
from the line-bunch $H$).
We can think of 
such a doca
$\B$ which has no reset control states and no zero rules
and arises from $\A$ as follows:
\begin{quote}
If $(p,a,1,p',j_1)$ and  $(q,a,1,q',j_2)$ are (positive) rules of $\A$,
where  $p',q'$ are stable,
and $L,L'$ are two lines from $H$ defined by valid slope-shift pairs
$(\mu,\tau)$, $(\mu,\tau')$, respectively,
and $j_2-\mu\cdot j_1=\tau'-\tau$

then $((p,q,L),a,1,(p',q',L'),j_1)$ is a (positive) rule of $\B$.
\end{quote}
An equivalent formulation of the condition $j_2-\mu\cdot j_1=\tau'-\tau$ 
is to say that
for all positive $m,n\in\N$ we have
$(m,n)\in L$ iff $(m{+}j_1,n{+}j_2)\in L'$
(i.e., $n=\mu\cdot m +\tau$ iff $n{+}j_2=\mu\cdot (m {+}j_1)+\tau'$).

For any tuple $(p,q,L,a)$ there is obviously at most one 
tuple $(p',q',L',j_1)$ such that 
 $((p,q,L),a,1,(p',q',L'),j_1)$ is a rule of $\B$; hence $\B$ is
 indeed a doca. The size of $\B$ (in particular 
 the number of control states of $\B$) is small
since the number of lines in $H$ is small
(recall Prop.~\ref{prop:fewlines}(1)).
 
 It is clear that
any positive path 
in $\T(\A)\times \T(\A)$ 
which visits only the pairs projected to the line-bunch $H$
corresponds to a path in $\T(\B)$;
the paths~(\ref{eq:lineclimb}) and~(\ref{eq:lineclimblinenot})  
illustrate this correspondence.

By Observation~\ref{prop:shortestsubpath}(1),
the path~(\ref{eq:lineclimb}) is a shortest path 
from $(p_0(m_0),q_0(n_0))$ to $(p_z(m_z),q_z(n_z))$
in $\T(\A)\times\T(\A)$.
By Lemma~\ref{lem:valiantpositpath},
a shortest path from $((p_0,q_0,L_0),m_0)$
to $((p_z,q_z,L_z),m_z)$ 
in $\T(\B)$  is of the form
 $((p_0,q_0,L_0),m_0)\gt{w} ((p_z,q_z,L_z),m_z)$
where 
$w=u_1v^iu_2$
for some short $u_1,v,u_2$ (short w.r.t. the size of $\B$ which is small)
and some $i\geq 0$; moreover, we can assume that 
the effect (the counter change) of
the respective control state cycle
$((p,q,L),..)\gt{v}((p,q,L),..)$ is positive (since $m_0<m_z$).

There is a slight problem that the path 
 $((p_0,q_0,L_0),m_0)\gt{w} ((p_z,q_z,L_z),m_z)$ in $\T(\B)$ might 
 not correspond to a positive path
 from $(p_0(m_0),q_0(n_0))$ to $(p_z(m_z),q_z(n_z))$
in $\T(\A)\times\T(\A)$ 
since 
$\B$ can go through a configuration $((p,q,L),m)$ where
$(\mu,\tau)$ is the slope-shift pair
of $L$ and $\mu\cdot m +\tau\leq 0$.
Nevertheless $u_1,v,u_2$ are short, 
and this problem thus cannot arise when $n_0$ is
larger than a small bound $\mathsf{b}$.
For showing that the path ~(\ref{eq:lineclimb}) is short, 
it suffices to show that its suffix starting in the first
$(p_j(m_j),q_j(n_j))$ where $n_j$ exceeds $\mathsf{b}$ is short.
(The prefix before such $(p_j(m_j),q_j(n_j))$ is obviously short.)

We thus immediately assume that $n_0$ is larger than $\mathsf{b}$,
which then allows us to assume that $a_1a_2\dots a_z$
in~(\ref{eq:lineclimb}) is $w=u_1v^iu_2$, as deduced from $\T(\B)$.  
We now perform a case analysis.

\begin{enumerate}
\item
$m_z=m_0{+}1$

By applying  Cor.~\ref{cor:valiantpositpath}
to the doca $\B$, we deduce that~(\ref{eq:lineclimb}) is short.

\item
$\lev(p_z(m_z),q_z(n_z))=0$

Path~(\ref{eq:lineclimb}) is short since 
$i\leq |u_1|+|u_2|+|v|$.
Otherwise 
by cutting off a copy of the cycle $v$, i.e. by performing
$u_1v^{i-1}u_2$ from $(p_0(m_0),q_0(n_0))$,
we would reach 
$(p_z(m_z{-}d_1),q_z(n_z{-}d_2))$
where $d_1$ is the effect of the cycle
$((p,q,L),..)\gt{v}((p,q,L),..)$ and 
$d_2=\mu\cdot d_1$ for the slope $\mu$ of $L$ (i.e. of the line-bunch
$H$).
We would 
thus reach
a pair with the zero eqlevel earlier 
(contradicting Observation~\ref{prop:shortestsubpath}(2)).

\item
There is an eqlevel-decreasing both-side reset step
$(p_z(m_z),q_z(n_z))\trans{a}(p'(0),q'(0))$
where
$p_z(m_z)\trule{a}s(m)\trule{\varepsilon} p'(0)$, 
$q_z(n_z)\trule{a}s'(n)\trule{\varepsilon} q'(0)$. 

Now $i\leq |u_1u_2v|+\per_{s}\cdot\per_{s'}$, 
since otherwise
by cutting off
$\per_{s}\cdot\per_{s'}$ copies of $v$ we would reach
$(p'(0),q'(0))$ earlier. Hence  (\ref{eq:lineclimb}) is short in this
case as well.

\item
There is an eqlevel decreasing simple step
$(p_z(m_z),q_z(n_z))\trans{a}(p'(m'),q'(n'))$ (as from $P_2$ 
in Fig.~\ref{fig:cyclesinlines}).

Then (``the diamond points in Fig.~\ref{fig:cyclesinlines}'', i.e.)
the sequence of pairs $(p'(m'_j),q'(n'_j))$ where
\begin{center}
$(p_0(m_0),q_0(n_0))\xrightarrow{u_1v^ju_2a}(p'(m'_j),q'(n'_j))$
\end{center}
and $j$ ranges over 
$|u_1u_2v|, |u_1u_2v|+1,|u_1u_2v|+2,
\dots, i{-}1, i\,$
is obviously a strange periodic sequence
(by recalling Observation~\ref{prop:shortestsubpath}(2)).
Since this sequence is short 
(by Prop.~\ref{prop:strangeshort}), also  (\ref{eq:lineclimb}) is
short.

\item
There is an eqlevel decreasing one-side reset step
$(p_z(m_z),q_z(n_z))\trans{a}(p'(m'),q'(0))$ (as from $P'_2$ 
in Fig.~\ref{fig:cyclesinlines}); we assume 
$q(n_z)\trule{a}s(n)\trule{\varepsilon} q'(0)$.

Then (``a subsequence of 
box points in Fig.~\ref{fig:cyclesinlines}'', namely)
the sequence of pairs $(p'(m'_j),q'(0))$
where 
\begin{center}
$(p_0(m_0),q_0(n_0))\xrightarrow{u_1v^ju_2a}(p'(m'_j),q'(0))$
\end{center}
and $j$ ranges over 
$i-x\cdot\per_s, i-(x{-}1)\cdot\per_s, i-(x{-}2)\cdot\per_s,
\dots,$
$i-2\cdot\per_s, i-\per_s, i$
where $x=(i{-}|u_1u_2v|)\div\per_s$ 
is obviously a strange periodic sequence.
Since this sequence is short 
(by Prop.~\ref{prop:strangeshort}), also  (\ref{eq:lineclimb}) is
short.
\end{enumerate}
\end{proof}

\subsection{Gaps in {\large {\bf $\elzc$}} 
are small}\label{subsec:gapssmall}

Assuming a doca $\A$, with the associated det-LTS $\LTSext(\A)$,
by Def.~\ref{def:elzc} we have
\begin{center}
$\elzc=\{e\in\N\mid$
there are two stable zero configurations $C,C'$ in $\LTSext(\A)$ s.t.
$C\depictlev{e}C'\}$.
\end{center}
We assumed $0\in\elzc$ and we fixed
an ordering $e_0<e_1<\cdots<e_f$ of $\elzc$.
We finally aim to contradict the existence 
of a large gap between $e_i=e_D$
and $e_{i+1}=e_U$ for some $i, 0\leq i<f$
(recall Fig.~\ref{fig:gapineqlevels});
this will finish a proof of
Theorem~\ref{th:strongpolywitness}. 

Before proving Lemma~\ref{prop:bnotminshort},
we sketch the idea informally, using Fig.~\ref{fig:downfromeU}.
Let us consider an eqlevel-decreasing path  
in $\LTSext(\A)\times\LTSext(\A)$, 
like~(\ref{eq:sinkeuseq}) below, which starts from
a pair  $(C_0,C'_0)$ of stable zero configurations satisfying
$\lev(C_0,C'_0)=e_U$;
let $(C_j,C'_j)$ be the pair  visited by our path 
after $j$ steps.
If both $C_0,C'_0$ are in $Q_\modp$
(recall that
$Q_\modp=\{\modstate(p(m))\mid p\in Q_\stable, m\geq 0\}$)
then also 
$C_1,C'_1$ are stable zero configurations (maybe in $\T(\A)$),
and thus $e_D=e_U{-}1$; the gap is really small in this case.
We thus further assume $C_0\not\in Q_\modp$ (hence $C_0=p(0)$ is in
$\T(\A)$);
this also handles the case  $C'_0\not\in Q_\modp$ by symmetry.

We are now not primarily interested in studying how 
the concrete pairs $(C_j,C'_j)$
can look like; we are interested in
the tuples $(b_j,\ell_j,r_j,o_j,d^L_j,d^R_j)$ 
associated with $(C_j,C'_j)$ by Def.~\ref{def:quadruple}
(recall Fig.~\ref{fig:quadruple}). 
The dependence of this tuple on $j$ is partly sketched in  
 Fig.~\ref{fig:downfromeU}.

\begin{figure}[ht]
\centering
\epsfig{scale=0.25,file=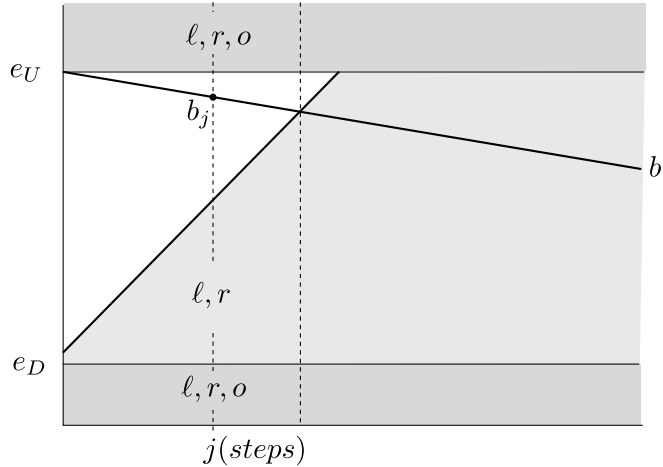}
\caption{
Constraints on $b_j,\ell_j,r_j,o_j$ after $j$
steps of 
an eqlevel-decreasing path 
with $b_0=e_U$}\label{fig:downfromeU}
\end{figure}

Since our path is eqlevel-decreasing
(the eqlevel drops by $1$ in each step), we know that
$b_j=e_U-j$, which is depicted by a line 
(in the standard sense, having nothing to do with IL-equality lines)
starting in point $(0,e_U)$ and having the slope $-1$.
(For a better overall appearence,  
the vertical unit length
in Fig.~\ref{fig:downfromeU}
is
smaller than the horizontal one.)

Each $o_j$ is either $\omega$ or an element of $\elzc$ 
(of $\textsc{E}_3$ after Def.~\ref{def:elzc}); in particular,
$o_j\geq e_U$ or $o_j\leq e_D$, which is depicted as a constraint in 
Fig.~\ref{fig:downfromeU}, using the horizontal lines at levels $e_U$
and $e_D$.

We now recall 
Prop.~\ref{prop:indlevel} and the fact that each finite $\il(q(0))$ is in
$\elzc$ (in $\textsc{E}_2$ after Def.~\ref{def:elzc}).
Hence for each $\ell_j$ we have either $\ell_j\geq e_U$ or 
$\ell_j\leq e_D+\rho_\textsc{m}\cdot j + \sigma_\textsc{m}$ where 
$\rho_\textsc{m}$
is the maximal number appearing as $\rho$
in the fixed equalities $\il(p(m))=\rho\cdot m+\sigma+e$, 
and $\sigma_\textsc{m}$
is the maximal number appearing there as $\sigma$.
(We use the fact that the counter value is at most $j$
in $C_j$, as well as in $C'_j$ when $C'_j$ is also in $\T(\A)$,
since we started from zero configurations.)
We recall that both $\rho_\textsc{m}$ and
$\sigma_\textsc{m}$ are small rational numbers.
The above constraints on $\ell_j$ are also depicted in 
Fig.~\ref{fig:downfromeU}, using the horizontal line at level $e_U$
and the line
starting in
$(0,e_D{+}\sigma_\textsc{m})$ and having the slope $\rho_\textsc{m}$.
The same constraints hold for $r_j$.

We note that if the horizontal coordinate of the 
intersection of the ``$b$-line'' (with slope
${-}1$) and the ``$\ell,r$-line'' (with the slope $\rho_\textsc{m}$)
is small then $e_U-e_D$ is small.
This is clear by noting that 
$b_j=e_U{-}j\leq e_D+\rho_\textsc{m}\cdot j+\sigma_\textsc{m}$
implies $e_U{-}e_D\leq (1{+}\rho_\textsc{m})\cdot
j+\sigma_\textsc{m}$. 

In fact, we will show even something stronger, namely that the maximal 
prefix of our path in which $b_j$ (for $j>0$) is ``solitary'', i.e. 
$b_j\not\in\{\ell_j,r_j,o_j\}$, is short.
This will be based on
Cor.~\ref{prop:cormineqlevel}, applied to the 
``rectangle'' $(b_j,\ell_j,r_j,o_j)$.
The previously established facts,
like that about few possible values $o_j$, 
will entail that in a long b-solitary prefix 
we would ``usually'' have $\ell_j=r_j<\omega$, 
which in turn would
entail a long line-climbing segment; this would contradict 
Prop.~\ref{prop:climbshort}.

\begin{definition}\label{def:bsolitary}
A \emph{pair} $(C,C')$ 
of stable configurations in $\LTSext(\A)$
with the associated eqlevel tuple
 $(b,\ell,r,o,d^L,d^R)$
is
\emph{b-solitary} if
$b\not\in\{\ell,r,o\}$. 
\\
A path in $\LTSext(\A)\times \LTSext(\A)$
is \emph{b-solitary} if each configuration-pair in the path is b-solitary.
\end{definition}

We note that in a b-solitary pair $(C,C')$
we must have that at least $C$ is in $\T(\A)$,
by our choice in
Def.~\ref{def:quadruple}.

\begin{lemma}\label{prop:bnotminshort}
All gaps $e_U{-}e_D$ in $\elzc$ are small.
\end{lemma}

\begin{proof}
We assume some $e_D,e_U\in\elzc$  
where $e_D<e_U$ and there is
no $e\in\elzc$ such that $e_D<e<e_U$,
and consider 
an eqlevel-decreasing path 
\begin{equation}\label{eq:sinkeuseq}
(C_0,C'_0)\gt{a_1}(C_1,C'_1)\gt{a_2}(C_2,C'_2)\gt{a_3}\cdots
\gt{a_z}(C_z,C'_z)
\end{equation}
in $\LTSext(\A)\times\LTSext(\A)$
where $C_0,C'_0$ are stable zero configurations, 
$C_0\depictlev{e_U}C'_0$, and $C_z\depictlev{0}C'_z$. 
We thus have $\eqlevel(C_j,C'_j)=e_U-j$ for all $j\in\{0, 1,\dots,z\}$.

Our aim is to show that   $e_U-e_D$ is small.
If $C_0\in Q_\modp$ and $C'_0\in Q_\modp$
then $C_1,C'_1$ are also (stable) zero configurations, and thus
$\eqlevel(C_1,C'_1)=e_U{-}1\in \elzc$; 
we thus have $e_U-e_D=1$.

We thus further assume that $C_0\not\in Q_\modp$ 
(while $C'_0\not\in Q_\modp$ is handled by symmetry).
Let $(b_i,\ell_i,r_i,o_i,d^L_i,d^R_i)$ be the eqlevel tuple associated
with $(C_i,C'_i)$ ($i=0,1,\dots,z$), as in Def.~\ref{def:quadruple};
in the case $C'_i\in Q_\modp$ we thus
have $r_i=\omega$, $b_i=d^L_i$, $o_i=d^R_i$.

We now note that if 
there is some small $j>0$ such that $(C_j,C'_j)$ is not b-solitary then
$e_U-e_D$ is small. This follows from the following two facts.
\begin{enumerate}
\item
If $b_j=o_j$ then $b_j=e_U-j\leq e_D$ 
(since $o_i$
belongs to $\elzc$ for all $i$);
hence $e_U-e_D\leq j$.
\item
If $b_j=\ell_j$ or $b_j=r_j$ then 
$b_j=e_U-j\leq e_D+\rho_\textsc{m}\cdot j+\sigma_\textsc{m}$,
and thus $e_U{-}e_D\leq (1{+}\rho_\textsc{m})\cdot
j+\sigma_\textsc{m}$, as was already discussed before
Def.~\ref{def:bsolitary}.

\end{enumerate}
We now fix $j$ so that 
\begin{equation}\label{eq:maxBsolprefix}
(C_1,C'_1)\gt{a_2}(C_2,C'_2)\gt{a_3}\cdots
\gt{a_j}(C_j,C'_j)
\end{equation}
is the maximal b-solitary prefix 
of the path~(\ref{eq:sinkeuseq}) in which the first step is removed.
We will show that $j$ is small, by which the proof will be finished;
we further assume $j\geq 1$.

The assumption $C_0\not\in Q_\modp$ implies
$C_1\not\in Q_\modp$ (hence $C_1=p(m)$ for some $p\in Q_\stable$ and some
$m\in\{0,1\}$).
Suppose $(C_1,C'_1)\gt{a_2}(C_2,C'_2)\gt{a_3}\cdots
\gt{a_{j_1}}(C_{j_1},C'_{j_1})$ is the maximal prefix 
of~(\ref{eq:maxBsolprefix})
such that
$C'_{j_1}\in Q_\modp$;
we put $j_1=0$ if $C'_{1}\not\in Q_\modp$.
For all $i\in\{1,2,\dots j_1\}$ we have $b_i=e_U-i$, $r_i=\omega$,
and $b_i\not\in\{\ell_i,r_i,o_i\}$;
hence
$\ell_i=o_i<b_i$ (by Cor.~\ref{prop:cormineqlevel}).
By Prop.~\ref{prop:ideafewc}(2) and Prop.~\ref{prop:conseqilpm}(1),
the set $\{C_i\mid 1\leq i\leq j_1\}$ is small,
 which implies that the set $\{b_i\mid 1\leq i\leq j_1\}=
\{d^L_i\mid 1\leq i\leq j_1\}$
 is small,
by Prop.~\ref{prop:ideafewcfix}(2).
Since $b_{i_1}\neq b_{i_2}$
if $i_1\neq i_2$, we get that $j_1$ is small.
It is thus sufficient to show that the suffix
\begin{equation}\label{eq:suffixmaxBsolprefix}
(C_{j_1+1},C'_{j_1+1})\gt{a_{j_1+2}}(C_{j_1+2},C'_{j_1+2})\gt{a_{j_1+3}}\cdots
\gt{a_j}(C_j,C'_j)
\end{equation}
of~(\ref{eq:maxBsolprefix}) is short. Let us rewrite 
(\ref{eq:suffixmaxBsolprefix}) as 
\begin{equation}\label{eq:pmsuffixmaxBsolprefix}
(p_0(m_0),q_0(n_0))\gt{a'_1}(p_1(m_1),q_1(n_1))\gt{a'_2}
\cdots
\gt{a'_{j'}}(p_{j'}(m_{j'}),q_{j'}(n_{j'}))
\end{equation}
where $j'=j-(j_1{+}1)$, and 
$(p_i(m_i),q_i(n_i))=(C_{j_1+1+i},C'_{j_1+1+i})$,
$a'_i=a_{j_1+1+i}$ for $i=0,1,\dots,j'$.
We note that $m_0+n_0$ is small (since $j_1$ is small and
$C_0,C'_0$ are zero configurations).

For simplicity, by 
$(b_i,\ell_i,r_i,o_i,d^L_i,d^R_i)$, where $0\leq i\leq j'$,
we further denote the eqlevel tuple associated 
with $(p_i(m_i),q_i(n_i))$ (not with $(C_i,C'_i)$ anymore).
Since the path~(\ref{eq:pmsuffixmaxBsolprefix}) is eqlevel-decreasing,
there is \emph{no repeat}, i.e.
$(p_{i_1}(m_{i_1}),(q_{i_1}(n_{i_1}))\neq (p_{i_2}(m_{i_2}),(q_{i_2}(n_{i_2}))$
if $i_1\neq i_2$.

For each $i\in\{0,1,\dots,j'\}$, the pair
$(p_i(m_i),q_i(n_i))$ is b-solitary, and thus
\begin{center}
$\min\{b_i,\ell_i,r_i,o_i\}$ 
is $r_i=o_i$ or $\ell_i=o_i$ or
$\ell_i=r_i$. 
\end{center}
We now aim to show that 
\begin{equation}\label{eq:fewilequal}
\text{there are only few } i\in\{0,1,\dots,j'\}
\text{ for which we do not have } \ell_i= r_i<\omega.
\end{equation}
To establish~(\ref{eq:fewilequal}),
it suffices to show that 
the sets $\{i\mid 0\leq i\leq j', r_i=o_i<\omega \}$ and
$\{i\mid 0\leq i\leq j', \ell_i=o_i<\omega \}$ are small;
by symmetry it suffices just to show that the former set is small.

We first note 
that the set
$$\{q_i(n_i)\mid 0\leq i\leq j', r_i=o_i<\omega \}$$
is small
by Prop.~\ref{prop:ideafewc}(2) 
and~\ref{prop:conseqilpm}(1).
Hence also the set $$\{d^R_i\mid 0\leq i\leq j', r_i=o_i<\omega\}$$
is small, by Prop.~\ref{prop:ideafewcfix}(2).
The set
$$\{i\mid 0\leq i\leq j',
r_i=o_i<\omega,
\min\{b_i,\ell_i,d^R_i\}=b_i=d^R_i\}$$ is thus
also small (recall that $b_{i_1}\neq b_{i_2}$
if $i_1\neq i_2$).
The set 
$$\{p_i(m_i)\mid 0\leq i\leq j',
r_i=o_i<\omega,
\min\{b_i,\ell_i,d^R_i\}=\ell_i=d^R_i\}$$
is also small, by recalling Prop.~\ref{prop:conseqilpm}(1).  
Since $\min\{b_i,\ell_i,d^R_i\}$ is $\ell_i=d^R_i$ or 
$b_i=d^R_i$ for all $i\in\{0,1\dots, j'\}$ 
(recall that $b_i=\ell_i$ is excluded in b-solitary pairs), we get
that both sets 
\begin{center}
$\{q_i(n_i)\mid 0\leq i\leq j', r_i=o_i<\omega \}$
and 
$\{p_i(m_i)\mid 0\leq i\leq j', r_i=o_i<\omega \}$
\end{center}
are small. 
Since there is no repeat in~(\ref{eq:pmsuffixmaxBsolprefix}), we get
that the set $\{i\mid 0\leq i\leq j', r_i=o_i<\omega \}$ is small.
We have thus established~(\ref{eq:fewilequal}). 

Let us now consider  the
\emph{sum-increasing subsequence} 
\begin{equation}\label{eq:sumicrease}
(p_{i_0}(m_{i_0}),q_{i_0}(n_{i_0})),
(p_{i_1}(m_{i_1}),q_{i_1}(n_{i_1})),
(p_{i_2}(m_{i_2}),q_{i_2}(n_{i_2})),
\dots
\end{equation}
of the sequence of pairs in~(\ref{eq:pmsuffixmaxBsolprefix}), 
where $0=i_0<i_1<i_2<\cdots$, and $i_{h+1}$ is the first such that 
$m_{i_{h+1}}+n_{i_{h+1}}$ is bigger than $m_{i_{h}}+n_{i_{h}}$
(for $h=0,1,2,\dots$).
If this subsequence is short then~(\ref{eq:pmsuffixmaxBsolprefix}) 
is obviously short since we started with small $m_0+n_0$
and $m_{i_{h+1}}+n_{i_{h+1}}\leq m_{i_{h}}+n_{i_{h}}+2$
(and there is no repeat in~(\ref{eq:pmsuffixmaxBsolprefix})).

For $h=0,1,2\dots$ we now consider the 
 subpaths
 of~(\ref{eq:pmsuffixmaxBsolprefix}) starting in
$(p_{i_h}(m_{i_h}),q_{i_h}(n_{i_h}))$
and finishing in
$(p_{i_{h+1}}(m_{i_{h+1}}),q_{i_{h+1}}(n_{i_{h+1}}))$; we call them
\emph{segments}.
A \emph{segment} is called \emph{unusual} if
\begin{itemize}
\item
the segment visits a pair $(p(m),q(n))$ such that 
$(m,n)$ is in no line-bunch, or is in the intersection of two different
line-bunches, or satisfies $m=0$ or $n=0$, or
\item
the segment contains a step
$(p(m),q(n))\gt{a}(p'(m{+}j_1),q'(n{+}j_2))$ such that 
$(m,n)$ and $(m{+}j_1,n{+}j_2)$ are in two different line-bunches.
\end{itemize}
Using~(\ref{eq:fewilequal}) and Prop.~\ref{prop:fewlines}
and the no-repeat property, we can
easily verify 
that there are only few unusual segments.

Any other segment, called \emph{usual}, is thus a positive path projected 
to one line-bunch;
moreover, the concatenation
of consecutive usual segments is also 
projected to one line-bunch.
We note that if $(p_{i_h}(m_{i_h}),q_{i_h}(n_{i_h}))$ and 
$(p_{i_{h'}}(m_{i_{h'}}),q_{i_{h'}}(n_{i_{h'}}))$, 
 for $h<h'$,
are in the same line
then $m_{i_h}<m_{i_{h'}}$.
Since there are only few lines, less than some small $\mathsf{b}_1$, 
and the lengths of eqlevel-decreasing line-climbing
paths are less than some small $\mathsf{b}_2$
by Prop.~\ref{prop:climbshort}, we cannot have more 
than $\mathsf{b}_1\cdot\mathsf{b}_2$ 
consecutive usual segments.
This finally implies 
that~(\ref{eq:sumicrease}) is short, and thus 
also~(\ref{eq:pmsuffixmaxBsolprefix}) is short.
Hence
$e_U-e_D$ is small.
\end{proof}

Now Lemma~\ref{prop:fewc} and Lemma~\ref{prop:bnotminshort}
give a proof of Theorem~\ref{th:strongpolywitness},
and thus also of Theorem~\ref{th:polywitness}.

\section{Additional remarks}\label{s:remarks}

The notions and their properties from the main
proof also help to answer related questions.
Here we only mention \emph{regularity}.
It is straightforward to verify that the language 
(the set of enabled traces) of a doca configuration $p(m)$ 
is non-regular iff we have
$p(m)\gt{u}q_1(n)\gt{v}q_2(n{+}k)\gt{w}q'(0)$
where $q_1(n)\gt{vw}q'(0)$ is a positive path and
$\il(q'(0))<\omega$.
(In this case, from $p(m)$ we can reach $q(n'))$ for some $q$ and
infinitely many $n'$ where $\il(q(n'))<\omega$.)
It is then a routine (though a bit technical)
to show that the regularity problem for doca
is in $\NL$ 
(and $\NL$-complete) as well.

\section*{Appendix (classical doca equivalence)}{\label{S App}}

\begin{figure}[ht]
\centering
\epsfig{scale=0.25,file=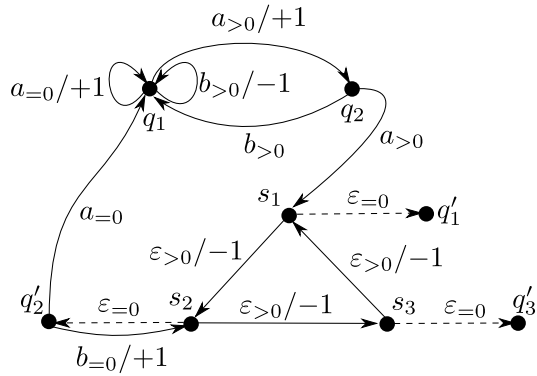}
\caption{
A classical doca}\label{fig:classicaldoca}
\end{figure}

The aim of this Appendix is to sketch the ideas of a routine reduction  
of the standard doca language equivalence problem to our 
$\docatrace$.
A classical definition would define
a \emph{doca} 
as a tuple
$\A=(Q,\Sigma,\delta,q_0,F)$ where
 $Q$ is a finite set of {\em control states},
 $\Sigma$ is a finite \emph{alphabet}, 
$\delta\subseteq Q\times (\Sigma\cup\{\varepsilon\})
\times\{0,1\}\times Q\times \{-1,0,1\}$ is
a {\em transition relation} 
satisfying the below given two conditions,
$q_0\in Q$ is the \emph{initial state}, and $F\subseteq Q$ is the set
of \emph{accepting states}.

In this context, $\varepsilon\not\in\Sigma$ is handled as a special symbol
but it plays the role of the empty word in the semantics.
The conditions for $\delta$ are the following.
\begin{enumerate}
\item
For each triple $(p,a,c)$, where $p\in Q$, 
$a\in\Sigma\cup\{\varepsilon\}$, $c\in\{0,1\}$ 
there is at most one  pair $(q,j)$ such that $(p,a,c,q,j)\in\delta$;
 moreover, $j\neq -1$ if $c=0$.
\item
If $(p,\varepsilon,c,q,j)\in\delta$ then there are no 
$a\in\Sigma$, $q'\in Q$, $j'\in\{-1,0,1\}$ such that 
$(p,a,c,q',j')\in\delta$.
\end{enumerate}
A {\em configuration} of $\A$ is a pair 
$(p,n)\in Q\times\N$; we write $p(n)$ instead of $(p,n)$, as previously.
We now define relations $\gt{w}$, $w\in\Sigma^*$, on $Q\times\N$ 
inductively as follows: 
$p(n)\gt{\varepsilon}p(n)$; 
if $(p,a,\sgn(n),q,j)\in\delta$ (where $a\in\Sigma\cup\{\varepsilon\}$)
then 
$p(n)\trans{a}q(n{+}j)$ (here $\sgn(n)=1$ if $n>0$ and 
$\sgn(n)=0$ if $n=0$); 
if $p(n)\gt{u}p'(n')$ and $p'(n')\gt{v}p''(n'')$ then 
$p(n)\gt{uv}p''(n'')$. Since the symbol $\varepsilon$ is handled as 
the empty word, we have $\varepsilon u=u\varepsilon =u$.
We define the \emph{language accepted by} $\A$ as
\begin{center}
$L(\A)=\{w\in\Sigma^*\mid q_0(0)\gt{w}q(n)$ for some 
$q\in F$, $ n\in\N\}$.
\end{center}
The \emph{language equivalence problem} asks, given two 
doca $\A_1$, $\A_2$ if $L(A_1)=L(A_2)$.

We now sketch the ideas of reducing this problem to our problem $\docatrace$.
First we note that we can take the disjoint union $\A$ of $\A_1,\A_2$ and
ask about the equality of languages of two different (initial)
configurations.
The doca $\A$, with $k$ control states,
can be routinely replaced 
by a doca $\A_{\shc}$ (with the ``Shrinked Counter''),
where a configuration $p(m)$ of $\A$ is represented by 
the configuration $p_i(j)$  of $\A_{\shc}$
where $i=m\bmod k$ and $j=(m\div k)$.
The control state set of $\A_{\shc}$ is $k$-times bigger,
to pay for shrinking the counter.

It is then easy to get rid of $\varepsilon$-rules which are not in
$\varepsilon$-cycles, and to get rid of $\varepsilon$-cycles 
with nonnegative effects. Finally, the only $\varepsilon$-rules which
remain are popping (decrementing the counter), and they are in cycles,
which is exemplified by the states $s_1,s_2,s_3$ 
in Fig.~\ref{fig:classicaldoca}.
To each such state $s$ in an $\varepsilon$-cycle we 
can add a control state $q_s$ with the
zero rule $(s,\varepsilon,0,q_s,0)$, to clearly separate the ``reset
control states'' from the ``stable ones''; this is illustrated
by $q'_1,q'_2,q'_3$ in Fig.~\ref{fig:classicaldoca}.
The final step of the transformation
to our reset-form doca (as in Fig.~\ref{fig:doca-example}) is now
obvious.
In the example, 
all  $s_1,s_2,s_3$ get the period $3$,
and we put $\goto_{s_2}(2)=q'_1$, 
$\goto_{s_3}(0)=q'_3$, etc.
(In fact, using $s_1$ is sufficient in our special case since 
the non-$\varepsilon$ incoming arcs of $s_2,s_3$ correspond to zero
rules only.)

Trace equivalence coincides with language equivalence when all states
are declared as accepting. 
A reduction from language equivalence to trace equivalence
can be sketched as follows.
For any triple $(q,a,c)$ such that $q\in Q_\stable$, $a\in\Sigma$,
$c\in\{0,1\}$ and there is no $(q,a,c,q',j)\in\delta$ we add the rule
$(q,a,c,q_\sink,0)$
where $q_\sink$ is an added ``sink loop'' state, with rules
$(q_\sink,a,c,q_\sink,0)$ for all $a\in\Sigma$ and $c\in\{0,1\}$.
We assume having arranged that all accepting control states are
stable, and we now add the ``loop'' rules $(q,a_\mathsf{acc},c,q,0)$ for 
a special fresh letter $a_\mathsf{acc}$ and all $q\in F$, $c\in\{0,1\}$
(so that $\lev(p(m),q(n))=0$ when $p\in F$, 
$q\not\in F$ or vice versa).

Since the reduction keeps the lengths of non-equivalence witnesses polynomially
related, the analogues of Theorems~\ref{th:polywitness}
and~\ref{th:nlcomplete} hold for the classical equivalence problem as
well.

\bibliographystyle{abbrv}
\bibliography{doca}

\end{document}